\documentclass[12pt]{amsart}
\usepackage[top=1.25in,bottom=1.25in,left=1.25in,right=1.25in]{geometry}
\usepackage[utf8]{inputenc} 
\usepackage{amsmath}
\usepackage{amssymb}
\usepackage{amsthm}
\usepackage{bbm}
\usepackage{graphicx}
\usepackage{float}
\usepackage{enumerate}
\usepackage{natbib}
\usepackage{multicol}
\usepackage{setspace,lipsum}
\usepackage{bm}
\usepackage{booktabs}
\usepackage[table]{xcolor}
\usepackage{float}

\usepackage{mathrsfs}
\usepackage{comment}
\usepackage{apptools}
\usepackage{url}
\newcommand{\RN}[1]{  \textup{\uppercase\expandafter{\romannumeral#1}}}
\providecommand{\eps}{\varepsilon}
\usepackage{booktabs}

\usepackage[subrefformat=parens]{subcaption}
\usepackage{mathtools}

\newtheorem{proposition}{Proposition}
\newtheorem{theorem}{Theorem}
\newtheorem{corollary}{Corollary}
\newtheorem{lemma}{Lemma}

\theoremstyle{remark}
\theoremstyle{remark}

\newtheorem{definition}{Definition}

\usepackage{etoolbox}
\makeatletter
\patchcmd{\@maketitle}
  {\ifx\@empty\@dedicatory}
  {\ifx\@empty\@date \else {\vskip3ex \centering\footnotesize\@date\par\vskip1ex}\fi
   \ifx\@empty\@dedicatory}
  {}{}
\patchcmd{\@adminfootnotes}
  {\ifx\@empty\@date\else \@footnotetext{\@setdate}\fi}
  {}{}{}
\makeatother

\makeatother

\AtAppendix{
\numberwithin{equation}{section}
\numberwithin{definition}{section}
\numberwithin{theorem}{section}
\numberwithin{proposition}{section}
\numberwithin{lemma}{section}
\numberwithin{corollary}{section}}

\DeclareMathOperator*{\argmin}{arg\,min}
\newcommand{\indep}{\perp \!\!\! \perp}

\usepackage{amsthm}
\newtheorem{example}{Example}[section]

\allowdisplaybreaks
\title{Artificial Intelligence Clones}
\author[Annie Liang]{Annie Liang$^*$}
\thanks{The author thanks Keaton Ellis, Drew Fudenberg, Andrei Iakovlev, David Levine,  Alexandru Lupsasca, Suproteem Sarkar, and Eduard Talamas for helpful comments, and National Science Foundation Grant SES-2145352 for financial support.}
\thanks{$^*$Northwestern University}
\date{\today}

\begin{document}

\maketitle

\begin{abstract}
Large language models, trained on personal data, are increasingly able to  mimic individual personalities. These ``AI clones'' or ``AI agents'' have the potential to transform how people search for matches in contexts ranging from marriage to employment. This paper presents a theoretical framework to study the tradeoff between the substantially expanded search capacity of AI representations and their imperfect representation of humans.  An individual's personality is modeled as a  point in $k$-dimensional Euclidean space, and an individual's AI representation is modeled as a noisy approximation of that personality. I compare two search regimes: Under in person search, each person randomly meets some number of individuals and matches to the most compatible among them; under AI-mediated search, individuals match to the person with the most compatible AI representation. I show that a finite number of in-person encounters yields a better expected match than search over infinite AI representations. Moreover, when personality is sufficiently high-dimensional, simply meeting two people in person is more effective than search on an AI platform, regardless of the size of its candidate pool. 

\end{abstract}

\section{Introduction}

Recent advances in large language models have brought us closer to a world in which individuals are represented by ``AI clones'' or ``AI delegates'' trained on their personal data \citep{Park2024GenerativeAgentSimulations}.  This technology has the potential to transform  how we search and match over human candidates, particularly in settings where direct engagement is costly. Labor markets are already moving in this direction: Roughly 90\% of employers already use an automated system to screen candidates \citep{amitabh2025aihiring}, with four in ten employers reporting in a survey that they intend to use AI to ``talk to'' candidates \citep{demopoulos2024jobapplicants}. These changes are facilitated by new AI platforms: Replicant enables job candidates to create digital representatives of themselves (see Figure \ref{fig:Replicant}), while Ribbon AI and Alex enable firms to automate their job recruiter \citep{rocha2025ai,temkin2025alex}. A similar shift is underway for dating markets. Bumble's founder has described a future in which ``your dating concierge could go and date for you with other dating concierges'' \citep{ForbesAustralia2024BumbleAIConcierge}, and platforms such as Volar Dating and Teaser AI have already adopted business models explicitly oriented around this approach \citep{Metz2024AIDating,lin2024grindr}.\footnote{This idea also appears in the Black Mirror episode ``Hang the DJ,'' in which the characters are revealed to be digital copies running simulations to gauge compatibility for their real-life counterparts.}$^,$\footnote{The rich medium of unstructured conversation allows AI representations to assess compatibility across many personality traits, in contrast to traditional online dating platforms which typically represent individuals by a limited set of attributes.} 

\begin{figure}[h]
    \centering
    \includegraphics[scale=0.2]{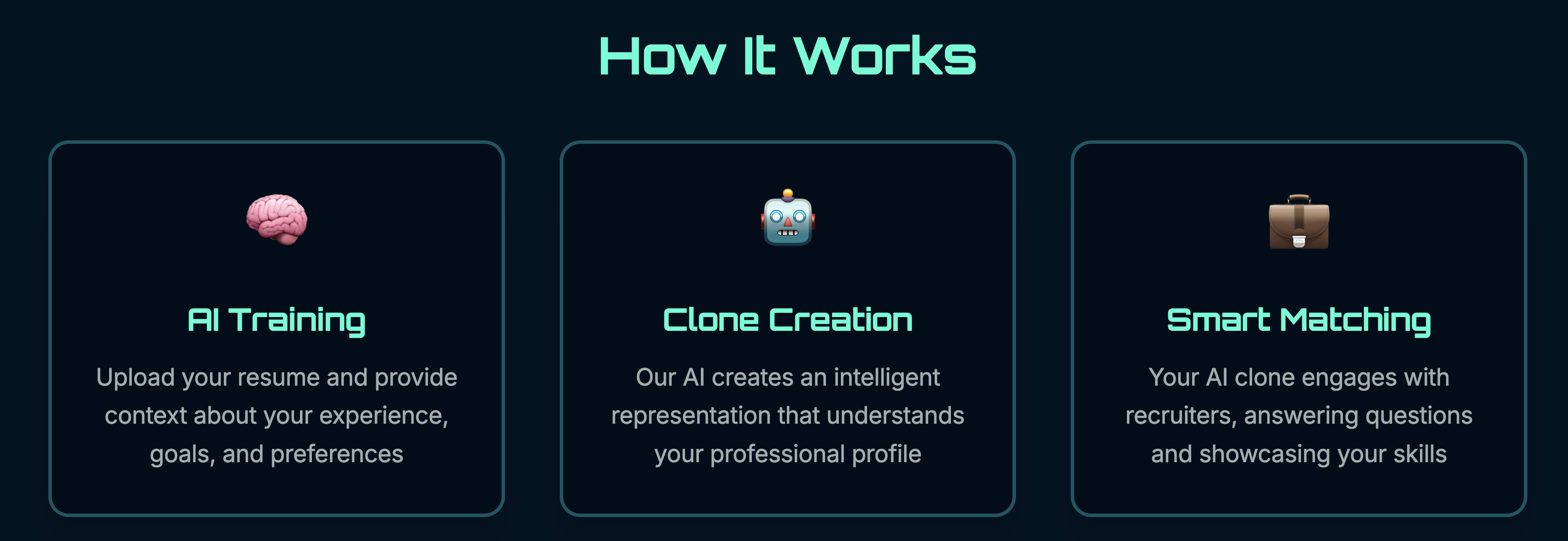}
    \caption{Replicant helps job candidates create AI clones to represent them in conversations with recruiters and potential employees. Source: Replicant website \citep{Replicant2024}.}
    \label{fig:Replicant}
\end{figure}

\noindent Across these domains, the appeal of AI-mediated matching is \emph{scale}: AI representations enable search over far more candidates than a person could feasibly meet. But representations of complex human personalities based on limited training data are necessarily imperfect. Thus, there is a tradeoff between the platform's improved scale of search and its reduced fidelity of representation.

This paper develops a theoretical framework to study this tradeoff. Individuals are modeled as points in $k$-dimensional Euclidean space, and AI clones as noisy representations of those points. I ask when search conducted over AI proxies can outperform direct human evaluation. The answer turns out to crucially depend on the dimensionality of the search space. When sufficiently many attributes are relevant to match compatibility ($k$ is large), the value of AI-mediated search is extremely low: participating on an AI platform, no matter how large its candidate pool, is no better than meeting two people in person.

The paper then turns to the broader social implications of use of AI representations for search. Specifically, I consider an extension in which the AI platform is able to develop more accurate AI representations for individuals in a ``data-rich'' group compared to a ``data-poor'' group. Despite the two groups being otherwise identical, Subject is matched to a data-rich individual with probability exceeding one-half. This probability moreover converges to 1 as the number of dimensions grows large. These results points to a potential new form of social stratification, where an individual's outcomes depend not just on their intrinsic characteristics, but also on how effectively artificial intelligence can understand and represent those characteristics.

Section \ref{sec:Model} describes the model. Subject's target match is normalized to the origin, with potential matches drawn uniformly at random from the unit ball. Subject prefers to be matched to an individual closer to their target, and Subject's payoff is decreasing in the Euclidean distance between the target and the match. I compare two search regimes: Under in-person search, Subject randomly meets $m$ individuals and matches to the closest among them. Under AI-mediated search, all individuals are represented by noisy perturbations of their true personality vector, henceforth their \emph{AI representations}.  Subject is matched to the individual whose AI representation is closest to the representation of Subject's target, but Subject's payoffs are determined by the true distance between this individual and the target.

This model is inspired by AI representations but applies broadly to search problems with two defining features. First, match quality depends on many attributes that are difficult to measure. Second, human evaluators have an intuitive and holistic understanding of match quality that does not require explicit assessment of those underlying attributes. For example, a recruiter can know that a candidate interviewed well without having precisely assessed the candidate on component dimensions. In such environments, automated search is constrained by imperfect representations of the candidates and the target, while human search is constrained by the cost of direct engagement. The central tradeoff is between the errors of automated search and the capacity constraints of human evaluation.

My results focus on a quantity I call the \emph{AI-equivalent sample size}, which measures the value of the AI platform. The AI-equivalent sample size is the smallest $m$ such that if Subject is permitted $m$ draws in the in-person regime, then  \emph{no} search size advantage on the AI platform can compensate for the error introduced by the AI representations. That is, Subject would rather meet $m$ individuals in person than search over an infinite number of individuals on the AI platform.

Section \ref{sec:MainResults} presents the main results. The first result, Proposition \ref{prop:Finite}, says that the AI-equivalent sample size is always finite. In other words, for any level of AI approximation error, there exists a finite number of in-person encounters sufficient to outperform any search size advantage offered by the AI platform. To show this, I demonstrate an intuitive but novel monotone-likelihood ratio property: smaller Euclidean distances between AI representations are associated with smaller Euclidean distances between the actual points. Thus the best-case match on the AI platform is an individual whose AI representation perfectly matches the representation of the target.  But because these representations imperfectly represent the underlying points, the expected actual distance between them is bounded away from zero. In contrast, under in-person search, the best of $m$ draws will eventually approach a perfect match, and thus sufficiently many draws must outperform the AI platform.
 
My second main result, Theorem \ref{thm:2}, says that when the number of dimensions of personality is large, then the AI-equivalent sample size is simply two. That is, in-person search yields a higher expected payoff so long as Subject searches over at least two people.

This theorem is a consequence of counterintuitive properties of high dimensional geometry: As the number of dimensions grows large, random points in the unit ball isolate away from one another and amass near the boundary. In the context of my model, this means that each individual becomes increasingly unique in the space of personalities. While the scarcity of good matches raises the value of searching in a large pool, the best match identified on the AI platform is increasingly driven by noise in the representation rather than true compatibility.  Theorem \ref{thm:2} shows that the breakdown of the platform's selection mechanism outweighs the benefits of its expanded search.

The final part of the paper considers an extension in which the AI platform has access to different quantities of data for different individuals.  Individuals are categorized into two groups, a ``data-rich'' and a ``data-poor'' group, which are differentiated by how well their AI representations mimic them. Despite the groups being  otherwise identical in characteristics, I show that Subject is matched to a data-rich individual with probability exceeding one-half. Moreover, this probability converges to 1 as either the number of dimensions, or the disparity in estimation errors, grows large. These results suggests that use of AI representations may systematically advantage individuals that the AI platform represents more accurately.

The remainder of the paper is organized as follows. Section \ref{sec:Model} describes the model. Section \ref{sec:MainResults} presents the main results. Section \ref{sec:Supplementary} presents the extension to heterogeneous individuals, and Section \ref{sec:Conclusion} concludes.

\subsection{Technological Context}
 \label{sec:TechContext}
 This paper is motivated by recent advances in artificial intelligence that make it possible for algorithms to function as personalized representations of specific individuals. Two developments are central to this shift. First, early evidence suggests that large language models reproduce general patterns of human conversation, and can be further adapted to a particular individual’s style when fine-tuned on their personal data. Second, unlike earlier machine-learning systems that were trained for narrowly defined prediction tasks (such as predicting medical conditions from an image scan), large language models  can engage in open-ended interaction both with humans and with one another. Together, these developments make large-scale automated search over human candidates technologically plausible. They are discussed in further detail below.

\subsubsection{Digital Representations of Specific Individuals}

Large language models are trained on large-scale corpora of text and other media drawn from the internet and related sources. While the first wave of such models adopted generic ``helpful assistant'' personas,\footnote{There is evidence that these systems already reproduce a wide range of human tendencies, preferences, and behavioral regularities \citep{pmlr-v202-aher23a,Mei2024Turing,Horton2023HomoSilicus}.} researchers and firms are increasingly exploring the possibility that such models can be fine-tuned on personal data---such as written text, social media posts, and recorded speech and video---to produce systems that mimic stable features of a \emph{specific} individual’s linguistic style, preferences, and judgment patterns.

Several firms already build and market such systems, often described as ``AI clones'' or ``digital twins'' of individuals. For example, Delphi AI enables public-facing figures to deploy AI agents that interact with audiences in their distinctive style; Replicant trains AI representatives for job candidates to participate in preliminary screening or interviews on their behalf; and Personal AI allows users to use their written communications to train a system that can, for example, write emails in their voice.

These generative agents are imperfect, but there are indications that they can already approximate some stable individual preferences and decision patterns with meaningful fidelity. For example, \citet{Park2024GenerativeAgentSimulations} find that LLM-based agents trained on short personal interviews reproduce individuals' General Social Survey survey responses and experimental choices with accuracy close to human test–retest reliability. As models improve and richer personal data become available, it is likely that such representations will continue to improve while never achieving perfect accuracy.

\subsubsection{Open-Ended Interaction and Scalable Search}
Large language models are capable of open-ended communication both with humans and with each other. For example, AI agents can interview one another, simulate collaboration, or engage in  conversation on behalf of the individuals they represent. The economic significance of these developments is that interaction itself becomes scalable. Platforms such as Ribbon AI and related `talent twin''systems are using conversational agents to conduct preliminary job interviews, and dating services including Volar Dating and Teaser AI are experimenting with AI-mediated exchanges in which users' agents converse to assess compatibility. Evaluation on these platforms is no longer constrained by human time or attention, allowing them to search across vastly larger pools of potential matches.

At the same time, digital representations are inevitably imperfect, and their fidelity may vary systematically across individuals as a function of data availability, online activity, and willingness to share personal information. Section \ref{sec:MainResults} establishes baseline limitations on AI-mediated search as a consequence of this imperfect representation.  Section \ref{sec:Supplementary} subsequently studies environments in which representation quality is heterogeneous and examines how such asymmetries shape outcomes.

\subsection{Related Literature}

This paper contributes to a growing literature on the social implications of artificial intelligence (AI), in particular to research comparing human and AI evaluation. AI predictions have been shown to outperform human experts across various prediction problems \citep{Kleinbergetal2017,rajpurkar2017chexnet,jung2017simple,Angelovaetal,Agarwaletal2023}. These papers all consider settings---such as medical diagnosis---where humans possess limited intuitive knowledge about the underlying decision problem.  By contrast, the present paper is motivated by settings where human actors possess a naturally rich understanding of a complex, subjective objective. This is a novel consideration relative to the literature, and the paper arrives at a very different conclusion regarding algorithmic versus human evaluation. 

There is an emerging literature on AI agents. This paper complements recent empirical work about the effectiveness of AI agents (such as \citet{JabarianHenkel2025VoiceAI} and \citet{Sarkar2025AIAgentsProductivity}) with a framework for studying the welfare implications of deploying AI agents in search.

Within economic theory, this work contributes to the rich literature on search, which has explored classic questions including how long to search for \citep{mccall1970economics,stigler1961economics,weitzman1979optimal}, what speed to search at \citep{UrgunYariv2024}, and where to search \citep{callander2011searching,Malladi2023}.\footnote{In related strategic experimentation models  (e.g., \citet{BoltonHarris1999} and \citet{KellerRadyCripps2005}), agents face uncertain payoffs and learn by repeatedly sampling an action---a feature absent from this paper.}  The present paper focuses on a new question regarding the role of the complexity of the search space, as measured by the number of attributes \citep{KLABJAN2014190}.  I show that dimensionality fundamentally alters optimal search behavior when search is conducted with error. This result differs from, for example, \citet{Bardhi} and \citet{Malladi2023}, who show that their characterizations of optimal search \emph{extend} but are not qualitatively changed in multiple dimensions.\footnote{In \citet{Bardhi} and and \citet{Malladi2023}, the searcher samples without error, and the focus is on inference about an unknown mapping from attributes to payoffs.}

This paper's analysis of the impact of complexity (as measured by the number of dimensions) relates to \citet{Ely2011} and \citet{FudenbergLevine2022}. \citet{Ely2011} shows that as an optimization space grows in complexity, decision-makers resort to incremental adaptations (``kludges'') which generate persistent inefficiencies. \citet{FudenbergLevine2022} examine how a partially-informed agent responds to systemic shocks, and show that the optimal intervention shrinks to zero as the number of dimensions grows large. My finding that higher dimensionality imposes fundamental limits on search aligns with these papers' high-level insights.

Key to Theorem \ref{thm:2} is a comparison of asymptotics as the dimensionality of the search space grows large. Rate of convergence results have a long history in economic theory, but typically involve limits in the quantity of information \citep{Vives,MoscariniSmith2002,LiangMu2020,FrickIijimaIshii2024} or the size of a population \citep{Harel2021,DasarathaGolubHak2023}. \citet{IakovlevLiang} consider a similar asymptotic to the present paper (namely, as the number of attributes grows large) but characterize limiting beliefs rather than the limiting value of search.

Finally, while this paper adopts the narrative of a dating platform, it diverges considerably both from the classic matching frameworks \citep{GaleShapley1962,RothSotomayor1992}, which examine how centralized mechanisms can achieve stable or efficient outcomes for  populations, and from the search and matching literature \citep{ShimerSmith2000,ChadeEeckhoutSmith2017}, which analyzes macroeconomic outcomes and equilibrium behavior in decentralized markets with many searchers. This paper instead focuses on the decision problem of a single agent navigating a high-dimensional search space.

\section{Model} \label{sec:Model}

\subsection{Setting} \label{sec:Setting}
Individuals are represented by vectors in $\mathbb{R}^k$, where each coordinate describes a  characteristic.\footnote{Example characteristics include intellectual curiosity, emotional sensitivity, relationship to time, relationship to authority, relationship to religion, energy bandwidth, style of humor, style of decision-making, and speed of decision-making.} A focal individual, hereafter \emph{Subject}, is searching for a match, which might be a spouse or an employee. Subject's ideal match or \emph{target} is normalized to the origin $x_0:=(0,0,\dots,0)$.  Potential matches $i \geq 1$ have characteristic vectors $X_i \in \mathbb{R}^k$ drawn uniformly at random from the $k$-dimensional unit ball $B_k \equiv \{ x \in \mathbb{R}^k : \| x \| \leq 1 \}$, which proxies for a first-round screen where individuals who are too far from Subject's target are not considered. Subject's utility from matching with individual $i$ is  \[u(X_i) = -\|X_i-x_0\| = -\|X_i\|,\] or the negative of the Euclidean distance between individual $i$ and the target.

There are two search regimes. Under \emph{in-person search}, Subject randomly meets $m$ individuals and matches to the most compatible among them.\footnote{The main model abstracts away from  whether the other individual would also agree to such a match; see Section \ref{sec:Discussion} for discussion.} The expected distance between the target and this individual is
\[d^{\text{IP}}_k(m) \equiv \mathbb{E}\left[\min_{1 \leq i \leq m} \| X_i \| \right].\]

Under \emph{AI-mediated search}, individuals do not interact in person but instead participate on a platform that uses  ``generators'' to stochastically represent the candidates and the target. Specifically, in the interaction between Subject and candidate $i$, Subject's target is represented by
\[Y_{0i} = x_0 + \varepsilon_{0i},\] 
and individual $i$ is represented by \[Y_{i0} = X_i + \varepsilon_{i0},\]
where $\varepsilon_{0i},\varepsilon_{i0} \sim \mathcal{N}(0, \sigma^2 I_k)$ are multivariate Gaussian noise terms that are independent of each other and across $i$. (The notation $I_k$ denotes the identity matrix in $k$ dimensions, and I assume throughout that $\sigma^2>0$.)  Figure \ref{fig:depiction} depicts this model in three dimensions, i.e., $k=3$.

\begin{figure}[h]
\begin{center}
    \includegraphics[scale=0.3]{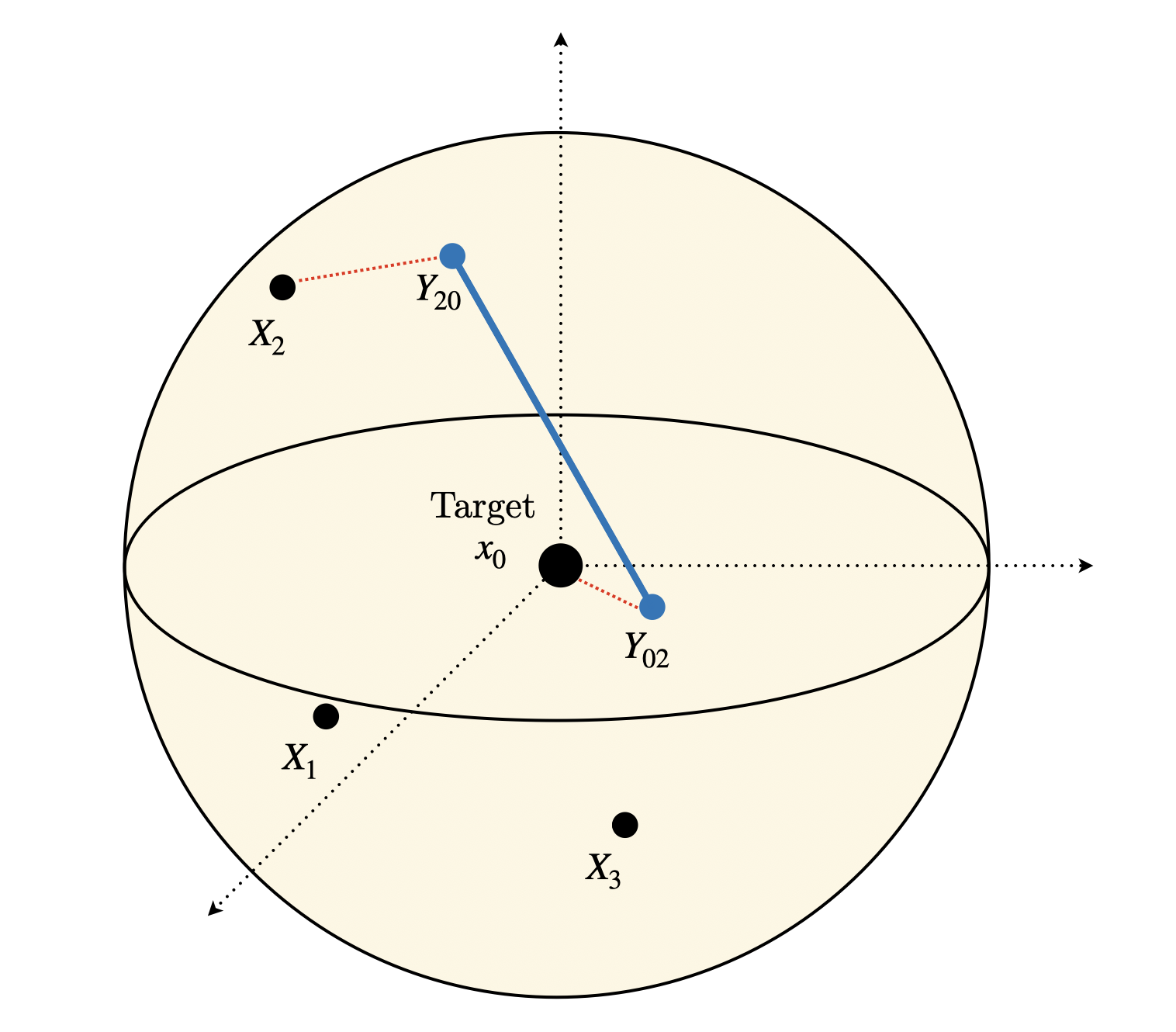}
\end{center}
\caption{Rather than observing $\|X_2-x_0\|$, the true distance between individual $i=2$ and the target, the platform observes $\|Y_{20}-Y_{02}\|$, the distance between their representations.} \label{fig:depiction} 
\end{figure}

There are $n$ candidates on the platform. For each candidate $i$, the platform observes $\|Y_{0i} - Y_{i0}\|$, the distance between the representation of the target and the representation of candidate $i$. The platform recommends the candidate whose representation is closest to the target representation:\footnote{This is the Bayesian-optimal selection for Subject given the platform's information, i.e., $i^*_n = \argmin_{1 \leq i \leq n} \mathbb{E}\left[\|X_i\| : \{\|Y_{0i} - Y_{i0}\|\}_{i=1}^n\right].$ See Lemma \ref{lemm:CondExp} for details.}
\[
i^*_n = \argmin_{1 \leq i \leq n} \|Y_{0i} - Y_{i0}\|.
\]
Crucially, although the selected candidate has the most compatible  \emph{representation}, Subject's payoffs are determined by the actual distance between Subject's target and this individual. For each sample size $n$, define
\[d^{\text{AI}}_{k}(n) \equiv \mathbb{E}\!\left[\|X_{i^*_n}\|\right]\]
to be the expected distance between the match $i^*_n$ and the target, in analogy to the previous $d^{\text{IP}}_k(m)$.

All results in this paper extend for two variations of this model. First, the representations may be fixed, so that $Y_{0}$ and $Y_{i}$ are generated once and fixed in each pairwise interaction. Second, AI representations may be employed on only one side of the market, i.e., only by the searcher or by the candidates. In each of these cases, the analysis extends with  minor modifications.\footnote{The first variation affects only the proofs of the results in Section \ref{sec:Supplementary}, which extend when we first condition on the common confounder $Y_{0}$. The second variation is handled by reducing the variance of the noise term when we write $Y_{i0}-Y_{0i}$ as a signal about $X_i-x_0$.}

\subsection{Value of AI-Mediated Search}
My main results characterize the following measure.

\begin{definition} Fix any number of dimensions $k$. The \emph{AI-equivalent sample size} $m^*_k$ is the smallest $m \in \mathbb{Z}_+$ satisfying 
\[d^{\text{IP}}_k(m)< d^{\text{AI}}_k(n) \quad \forall n \in \mathbb{Z}_+.\]
If the inequality above is not satisfied for any finite $m$, then $m^*_k=\infty$.
\end{definition}

The AI-equivalent sample size is the smallest integer $m$ such that searching over $m$ individuals in person leads (in expectation) to a better match than participating on the AI platform, \emph{no matter} the size of the platform's candidate pool. It serves as a quantitative benchmark for the value of AI-mediated search: An infinite AI-equivalent sample size means that for every $m$, one can choose $n$ large enough such that search on an AI platform with $n$ candidates outperforms in-person search over $m$ candidates. A finite AI-equivalent sample size means that the advantage of the AI platform is fundamentally capped, and a sufficiently thorough in-person search will outperform it no matter its size. In this case, the practical value of the platform depends on the size of $m_k^*$.

\subsection{Interpretation of Model}

The proposed framework applies to environments in which AI representations are already beginning to be deployed (Section \ref{sec:TechContext}), such as the following.

\begin{example}[Hiring]
A firm seeks to hire a suitable candidate. Under in-person search, a hiring manager interviews a small number of applicants. Under AI-mediated search, an AI recruiter interviews the AI representations of a substantially larger number of candidates. As discussed earlier, all results extend if an AI representation is only used on one side of this market, i.e., if an AI recruiter interviews humans or if a human recruiter interviews AI representations of candidates.
\end{example}

\begin{example}[Dating]
An individual seeks a partner. Under in-person search, the individual chooses the (truly) most compatible from a limited number $m$ of candidates. Under AI-mediated search, a platform first seeks to understand the individual's ideal match, and then evaluates a potentially much larger set of candidates $n \gg m$ relative to its understanding of the searcher's target. In the special case where the searcher prefers someone with similar characteristics, we can interpret the target representation as a representation of the searcher themselves.
\end{example}

The framework also extends naturally to settings where AI representations of individuals are not yet in use but plausibly will be, such as clients seeking lawyers, founders seeking investors, or  families seeking care providers.

Finally, although this paper is motivated by emerging uses of AI representations, the framework applies more broadly to search problems with two key features. First, match quality depends on many attributes that are individually difficult to measure. Second, the human searcher does not need to explicitly identify or evaluate each relevant attribute, but instead has an intuitive grasp of what constitutes a better or worse outcome.\footnote{Reinforcement learning with human feedback \citep{christiano2017deep} emerged as a crucial innovation for large language models precisely because it provides a way to train systems on judgments of this kind.} For example, a hiring manager can judge that an interview went poorly without having explicitly evaluated the candidate's communication style, problem-solving approach, or cultural fit. An AI platform, by contrast, lacks direct access to this holistic evaluation and must instead infer quality from noisy measurements of component attributes.\footnote{In the model, the assumption $\sigma^2>0$ captures this imperfect measurement and rules out settings with objectively observable attributes.} These features appear in other search and delegation problems as well, such as the following.

\begin{example}[Purchasing a Home] \label{ex:Home}
A family seeks to purchase a home. Under in-person search, the family visits \(m\) houses directly. Alternatively, the family imperfectly communicates its preferences to a system that evaluates a larger set \(n\) of homes relative to its understanding of the target. 
\end{example}

\section{Main Results} \label{sec:MainResults}

Section \ref{sec:Results} presents two results about the limits of AI-mediated search. Section \ref{sec:ProofSketch} outlines the proof of these results. Section \ref{sec:Discussion} discusses possible extensions.

\subsection{Limits of AI-Mediated Search} \label{sec:Results}
My first result says that the value of the AI representation regime is inherently capped: For every number of dimensions $k$, there is some finite number of in-person encounters that yields a better expected match than the AI platform, no matter how large its candidate pool.

\begin{proposition} \label{prop:Finite} For every number of dimensions $k \in \mathbb{Z}_+$, the AI-equivalent sample size $m^*_k$ is finite.   
\end{proposition}

Intuitively, imperfect representations place a fundamental limit on the quality of search: even with an arbitrarily large platform, the expected true distance between the Subject’s target and the platform’s best-ranked match remains bounded away from zero. In contrast, under in-person search, Subject's expected distance to the best among  $m$ individuals vanishes as $m$ grows large. Thus for sufficiently large $m$, Subject's expected match must be better in the in-person regime. 

Proposition \ref{prop:Finite} holds for any number of attributes $k$. The next result says that when $k$ is sufficiently large---so that many attributes are relevant to match quality---the value of AI-mediated search drops sharply.

\begin{theorem} \label{thm:2} For all $k$ sufficiently large, the AI-equivalent sample size is $m^*_k = 2$.
\end{theorem}
\noindent Thus, in-person search yields a better expected outcome so long as Subject searches over at least two people. The threshold dimensionality 
$k$ at which this reversal occurs depends on the magnitude of AI measurement error $\sigma^2$, but it remains finite whenever $\sigma^2>0$.

The implication is not that AI-mediated search is inherently ineffective, but rather that its usefulness depends critically on the effective dimensionality of match quality. When compatibility depends on a small number of characteristics, large-scale automated search can perform well. By contrast, when many distinct attributes matter, Theorem~\ref{thm:2} implies a fundamental limitation to AI-mediated search: Even a very small amount of direct, in-person evaluation yields a better expected outcome than participation on an arbitrarily large AI platform. (This conclusion relies on the relevant attributes not being well approximated by a lower-dimensional set; otherwise, despite many nominal attributes, the effective dimensionality of evaluation would be small.)

The number of relevant attributes $k$ likely varies across applications. In routine jobs, for instance, match quality may be predictable from a relatively small set of attributes. Other roles, such as selecting a CEO, plausibly involve a substantially larger set of traits. Long-term partner choice may involve an even richer set of relevant attributes \citep{Joel2017,Joel2020}---spanning values, personality, emotional compatibility, habits, and life goals---though the effective dimensionality may also vary across individuals. Theorem \ref{thm:2} thus implies that AI-mediated search can be valuable when alignment on a small number of attributes suffices, but that human judgment is indispensable---and ultimately the limiting constraint---when compatibility depends on a diverse, multidimensional set of attributes.

\subsection{Proof Sketch} \label{sec:ProofSketch}

The proof of Theorem \ref{thm:2} proceeds in two main steps.

\subsubsection{A uniform lower bound for $d^{\text{AI}}_k(n)$.} \label{sec:MLRP} Consider an idealized AI platform that continuously samples candidates until encountering one whose representation exactly matches the target representation, i.e., an individual $i$ for whom $\|Y_{i0}-Y_{0i}\|=0$. Let
\[
d^{\mathrm{AI}}_k(\infty) \equiv 
\mathbb{E} \left(\|X_i\| : \|Y_{i0}-Y_{0i}\|=0 \right)
\]
be the expected true distance to such a perfectly matched individual. The following lemma shows that this quantity is a uniform lower bound on $d_k^{\text{AI}}(n)$ across all $n$.

\begin{lemma} \label{lemm:LowerBound}
For all $n \in \mathbb{Z}_+$,
\[
d^{\mathrm{AI}}_k(\infty) \le d^{\mathrm{AI}}_k(n).
\]
\end{lemma}

To show this, I demonstrate that the pair of random variables
\[
\bigl(\|X_i\|,\ \|Y_{i0} - Y_{0i}\|\bigr)
\]
satisfies the monotone likelihood ratio property. That is, conditioning on a smaller representation distance $\|Y_{i0} - Y_{0i}\|$ shifts the conditional distribution of $\|X_i\|$ toward smaller values.\footnote{Given the parametric assumptions of Section~\ref{sec:Model}, it is well known that the random vector $(X_i, Y_{i0} - Y_{0i})$ satisfies the monotone likelihood ratio 
property (MLRP), and that $(h(X_i), g(Y_{i0} - Y_{0i}))$ also satisfies MLRP for any 
increasing functions $h$ and $g$. However, proving the specific MLRP relation for 
$\left(\|X_i\|, \|Y_{i0} - Y_{0i}\|\right)$ requires a novel argument because norms in 
$\mathbb{R}^k$ are not monotone in each coordinate, and thus standard 
coordinate-wise MLRP arguments do not apply.}  Consequently, the function
\[
s \mapsto \mathbb{E}\!\left[\,\|X_i\| \,\middle|\, \|Y_{i0} - Y_{0i}\| = s \right]
\]
is increasing in $s$. Since the selected individual must satisfy $\|Y_{i0} - Y_{0i}\| \ge 0$, the claim follows.

\subsubsection{Benchmark comparison}

Let $d_k^{\mathrm{IP}}(1)$ denote the expected distance to a single uniformly drawn individual from the unit ball. It is well known that
\[
d_k^{\mathrm{IP}}(1) = \frac{k}{k+1}.
\]
I show that both $d_k^{\mathrm{IP}}(2)$ and $d^{\mathrm{AI}}_k(\infty)$ converge to $d_k^{\mathrm{IP}}(1)$ as $k$ grows large, but crucially they do so at different rates. Specifically, the expected reduction in distance gained by picking the smaller of \emph{two} draws rather than one is vanishing on the order of $1/k$, 
\[
 d_k^{\mathrm{IP}}(1) - d_k^{\mathrm{IP}}(2)
  = \Theta\!\left(\frac{1}{k}\right) \]
while the improvement from participating on a completely saturated AI platform vanishes \emph{strictly faster} than $1/k$,
\[ d_k^{\mathrm{IP}}(1) - d_k^{\mathrm{AI}}(\infty) = o\!\left(\frac{1}{k}\right).\footnote{More precisely, the ratio $\frac{ d_k^{\text{IP}}(1) - d_k^{\text{AI}}(\infty)}{1/k}$ vanishes to zero, while the ratio $\frac{d_k^{\text{IP}}(1) - d_k^{\text{IP}}(2)}{1/k}$ is asymptotically a positive constant.}\]
Thus $d^{\text{AI}}_k(\infty)$ exceeds $d^{\text{IP}}_k(2)$ for large $k$, meaning that in sufficiently high-dimensional personality spaces, Subject expects to be closer to the best of two draws in the in-person regime than to the best match among an infinite number of candidates on the AI platform. Together with Lemma \ref{lemm:LowerBound} and an argument that $m_k^*>1$, this yields the desired result. 

\subsubsection{Intuition}

It might seem surprising that $d^{\mathrm{AI}}_k(\infty)$ converges to $d^{\mathrm{IP}}_k(1)$ as $k$ grows large. The underlying force is a basic geometric feature of high-dimensional spaces. When points are drawn uniformly from the unit ball in $\mathbb{R}^k$, almost all probability mass concentrates in a thin shell near the boundary, and pairwise distances between points rapidly concentrate around their expectation. In the context of this model, this means that as $k$ grows large people are increasingly unique, and compatibility with the target is increasingly hard to find. This does not mean that close points fail to exist---on a fully saturated platform, such points appear with probability one. But the platform must identify these rare compatible candidates using noisy measurements. As a result, although $\|Y_{i0} - Y_{0i}\|$ remains informative about true compatibility in expectation (Section~\ref{sec:MLRP}), the identity of the candidate minimizing this quantity is increasingly determined by noise rather than by genuine differences in compatibility. The platform’s ranking mechanism deteriorates, and its performance eventually converges to that of a single random draw.

\subsubsection{Simulation}
Although a bound for how large $k$ must be for Theorem \ref{thm:2} to apply is beyond the scope of this paper, Figure \ref{fig:Simulate} reports estimates for $d^{\text{IP}}_k(2)$ and $d^{\text{AI}}_k(N)$ with $N=10,000$ (chosen to be some arbitrary large number) as we vary the number of dimensions $k$.\footnote{In more detail: in each of 1000 iterations, I draw $N=10,000$ realizations of $(X_i,Y_{i0},Y_{0i})$ in the AI representation regime and find the index $i$ where $\| Y_{i0}-Y_{0i} \| $ is minimized. I then average over the values $\|X_{i^*}\|$ to derive an estimate of $d^{\text{AI}}_k(N)$. Likewise in the in-person regime, I draw 2 realizations of $(X_1,X_2)$ and average over the smaller of the two norms over the 1000 trials.} The estimate of $d^{\text{AI}}_k(N)$ exceeds the estimate of $d^{\text{IP}}_k(2)$ for $k\geq 150$. Thus if match compatibility is based on 150 or more attributes, then Subject should prefer an in-person search over two individuals over participating on an AI platform with 10,000 individuals. (See Table \ref{tab:Data} in Appendix \ref{app:Table} for the exact values in Figure \ref{fig:Simulate}.) Note that the small differences between $d^{\text{IP}}_k(2)$ and $d^{\text{AI}}_k(N)$ for $k$ large do not mean that the differences between in-person search and AI-mediated search are small, since in practice we expect searchers to search over more than 2 people. (Moreover, the size of the search may be endogenously selected as a function of $k$.)

\begin{figure}
    \begin{center}
        \includegraphics[scale=0.25]{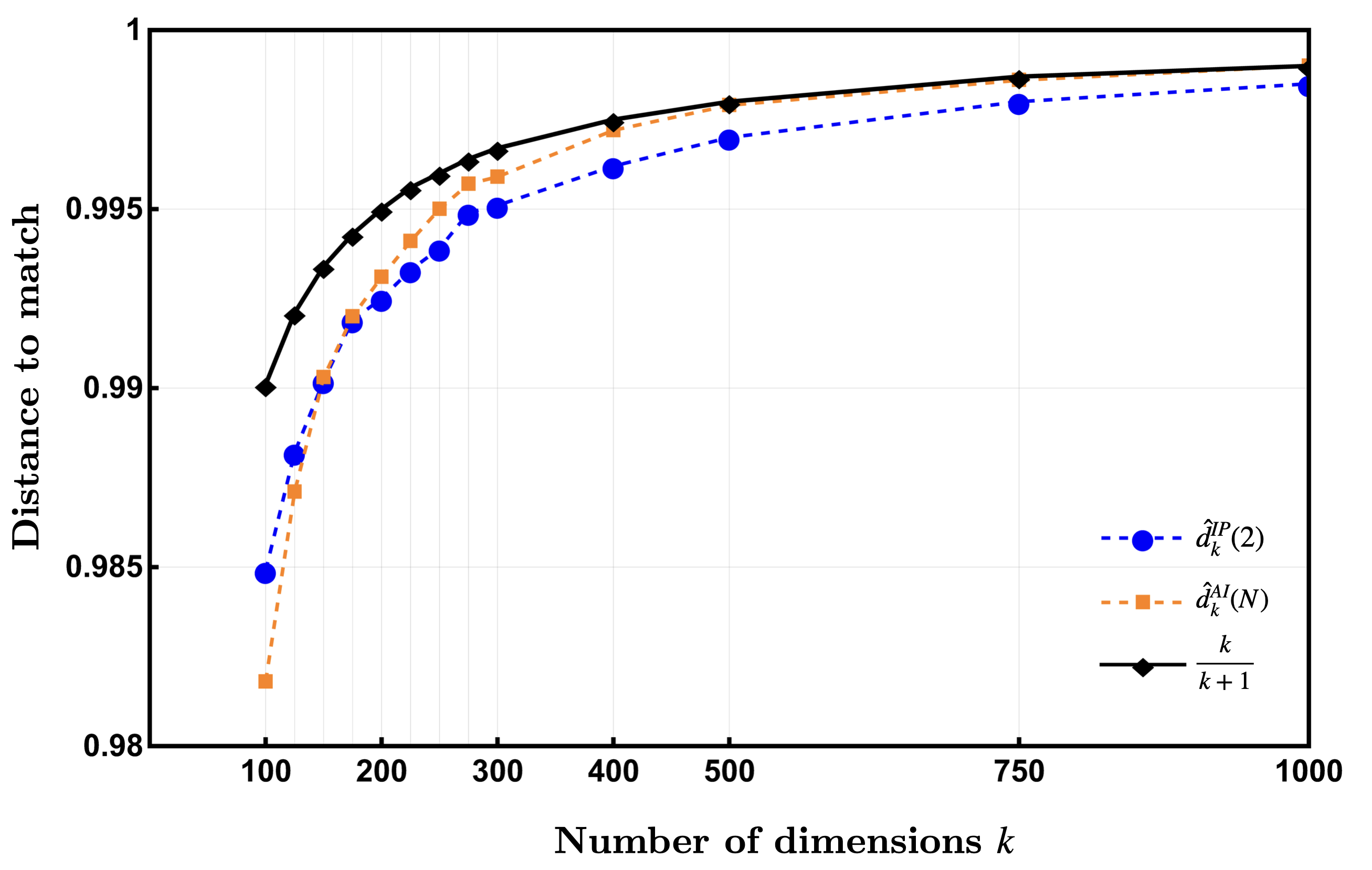}
        \caption{This figure reports estimates of $d^{\text{IP}}_k(2)$ and $d^{\text{AI}}_k(N)$ for $\sigma=0.05$ and $N=10,000$, using their average values over 1,000 draws. The quantity $d^{\text{AI}}_k(N)$ is initially below $d^{\text{IP}}_k(2)$ and eventually overtakes it. Both $d^{\text{IP}}_k(2)$ and $d^{\text{AI}}_k(N)$ converge to $k/(k+1)$ from below.} \label{fig:Simulate}
    \end{center}
\end{figure}

\subsection{Discussion} \label{sec:Discussion}

Several aspects of the model are discussed below.

\subsubsection{The role of the parametric assumptions.} 
In this model, personalities are uniformly distributed within a unit ball, and AI errors follow a multivariate Gaussian distribution. The proofs of the main results do depend on these parametric assumptions---both Proposition~\ref{prop:Finite} and Theorem \ref{thm:2}  rely on the  monotone-likelihood ratio property discussed in Section \ref{sec:MainResults} (Lemma \ref{lemm:CondExp}), and Theorem \ref{thm:2} further uses rate-of-convergence results that rely on exact expressions. Nevertheless, the underlying mechanisms are more general. The key insight behind Proposition~\ref{prop:Finite} is that the expected distance to the best match on the AI platform remains bounded away from zero (because of AI approximation errors). Meanwhile, Theorem~\ref{thm:2} is driven by the tendency of random points to isolate from one another in high dimensions, which causes their pairwise distances to concentrate around their expected values. This phenomenon persists under a wide range of distributional assumptions.

\subsubsection{Two-sided matching} 
This paper's model of search is one-sided, and thus abstracts from the question of whether Subject’s best match (once identified) would agree to this pairing. Although I do not pursue a complete strategic model in the present paper, in many reasonable formulations of such a model, Subject's expected payoff would decrease in the number of candidates considered by Subject's potential partners. This effect would further favor the in-person regime, where each individual searches over fewer candidates.

\subsubsection{Increasingly accurate AI}
Section \ref{sec:MainResults} reveals a limitation of the AI representation regime that persists regardless of the quality of AI approximation (i.e., the size of the noise parameter $\sigma$). Nevertheless, we might expect that as the AI representations become more accurate representations of their underlying users, the AI representation regime becomes more attractive relative to the in-person regime. In the appendix, I show that for every fixed dimension $k$, the quantity $d^{\text{AI}}_{k,\sigma}(\infty)$ (where the noise level $\sigma$ is now made explicit) converges to zero as $\sigma\rightarrow 0$. That is, Subject's expected distance to Subject's best match  converges to zero as the AI's approximation errors vanish. This means that while the value to increasing the size of the AI platform is capped, the value to increasing the accuracy of the AI representations is not.

\subsubsection{Sequential Search} Appendix \ref{app:Sequential} describes a sequential search model in which the sample sizes $m$ and $n$ are endogenously chosen by Subject. (It is possible to mirror the initial setting by supposing that the costs of sampling in the AI regime are lower, thus micro-founding the consideration of substantially larger $n$ on the AI platform.) I show that when there is a fixed cost to participating in the AI platform that is at least as large as the cost of searching over two people in person, then Subject prefers in-person search. 

\subsubsection{Human-AI Integrated Search}

One might consider using the AI platform to identify $m$ potential candidates, whom Subject then meets in person. This is unlikely to substantially improve Subject's payoffs beyond what Subject can achieve by searching over $m$ individuals. The logic is as follows: In the proof of Theorem \ref{thm:2}, I showed that (for large $k$) the best match on the AI platform is similar in distribution to a draw from a uniform distribution over the unit ball. By extension, I conjecture that (for large $k$) the joint distribution of the $m$ best matches on the AI platform converges to that of $m$ draws from a uniform distribution over the unit ball. This would mean that for large $k$, Subject prefers searching over $m+1$ people in person rather than selecting from $m$ people sourced from the AI platform.

\section{ Heterogeneous Data Quantities} \label{sec:Supplementary}

The remainder of the paper explores heterogeneity in the quality of the AI representations. In particular, since not all individuals have the same quantity of personal data to share (e.g., because of variation in social media usage) or the same willingness to share it (e.g., because of variation in privacy preferences), the ability of the AI platform to generate representations that accurately represent each person will differ. 

In this extension, individuals belong to either of two groups: a ``data-rich'' group for whom the approximation error is smaller, and a ``data-poor'' group for whom the approximation error is larger. Proposition \ref{prop:Noise} says that in this world we will see an inequality emerge between people who are better and worse understood by the AI. In particular, although the groups are otherwise identical in distribution---and thus, Subject's actual best match is equally likely to be from either group---the probability that Subject is matched to someone in the data-rich group strictly exceeds $1/2$. This probability further converges to 1 as either the number of dimensions grows large (Corollary \ref{prop:LimitNoiseLargek}) or the ratio of noise variances grows large (Corollary \ref{prop:LimitNoiseLargeVar}).

In more detail, the AI platform searches over $n$ individuals in each of two groups $g \in \{R,P\}$. Individual $i$ from group $g$ has true personality vector $X_i^g$, which is independently drawn from a uniform distribution on the unit ball. When the platform simulates an interaction between Subject and this candidate, the representation of individual $i$ is
\[Y^g_{i0}=X^g_i+\varepsilon_{i0}^g, \quad \quad \varepsilon_{i0}^g\sim \mathcal{N}(0,\sigma_g^2 I_k)\]
 and the representation of Subject's target is
\[Y^g_{0i} = x_0 + \varepsilon^g_{0i}, \quad  \quad \varepsilon^g_{0i} \sim \mathcal{N}(0,\sigma_0^2 I_k)\]
for some constant $\sigma_0^2>0$.\footnote{For interpretation, we might set $\sigma_0^2 \in\{\sigma_R^2,\sigma_P^2\}$, although this is not necessary for the result.} Crucially  $\sigma^2_R < \sigma^2_P$, so representations are more accurate for data-rich individuals. All noise terms are mutually independent.

As before, the AI platform identifies the individual whose AI representation is most compatible with Subject's AI representation. Subject is thus matched to individual $i^*_{n,k}$ from group $g^*_{n,k}$ where
\[(i^*_{n,k},g^*_{n,k}) = \argmin_{g\in \{R,P\}, 1 \leq i \leq  n}\| Y^g_{i0}-Y^g_{0i} \|.\]

The next result says that as the population grows large, the probability with which Subject is matched to someone from the data-rich group exceeds $1/2$.

\begin{proposition} \label{prop:Noise} In the limit as the population grows large, the probability that Subject's match is from the data-rich group converges to
\[\lim_{n\rightarrow \infty} P\left(g_{n,k}^* =R\right) > \frac12.\]
\end{proposition}

 Thus, even though  individuals from the data-rich group and data-poor group are identical in distribution (both are drawn uniformly at random from the unit ball), a data-rich individual is more likely to be identified as Subject's best match.\footnote{
 This result is reminiscent of \citet{CornellWelch1996}'s finding that more accurate signals about one group lead an employer to hire from that group with probability converging to 1 as the population grows large. But because my analysis focuses on small bounded distances rather than large unbounded qualities, the limiting probabilities here are interior.}

For a brief intuition, observe that each individual's distance to Subject is given by $\| Y^g_{i0} - Y^g_{0i} \|$. In a large population, Subject's best match is an individual for whom this distance is very small, i.e., an extreme draw from the underlying distribution. We thus need to  analyze whether that extreme is more likely to come from the data-poor or data-rich group. I show that (a suitable transformation of) the smallest distance to any group-$g$ individual, $M_n^g = \min_{1\le i \le n} S_i^g$, is asymptotically distributed like an exponential random variable with rate proportional to the density of $ Y^g_{i0} - Y^g_{0i}$ at zero---call this $f_k^g(0)$. This means that group $g$’s probability of ``winning,'' i.e., having the smallest minimum distance to Subject, is asymptotically determined by the relative size of its exponential rate, and thus by the size of $f_k^g(0)$. Since lower noise variance leads to a more sharply peaked density at zero (i.e., $f_k^R(0)>f_k^P(0)$), the data-rich group is more likely to yield the closest representation.

The following corollaries show that this inequality is exaggerated in either of two limits:  as the number of dimensions $k$ grows large, or the disparity between noise variances $\sigma_P^2/\sigma_R^2$ grows large. In both cases, the probability that a data-rich individual is selected converges to 1. 

\begin{corollary}[Many Dimensions] \label{prop:LimitNoiseLargek} In the limit as the number of dimensions and population size both grow large, the probability that Subject's match is from the data-rich group converges to 1, i.e.,
\[\lim_{k \rightarrow \infty} \left(\lim_{n\rightarrow \infty} P\left(g_{n,k}^* = R \right) \right)=1.\]
\end{corollary}

This result first takes the platform's population size $n$ to grow large for a fixed number of dimensions $k$, and then allows $k$ to increase. This order of limits is technically convenient, since it allows us to invoke the limiting characterization in Proposition \ref{prop:Noise}.\footnote{If we were to reverse these limits---taking $k$ to infinity first---there would be no single, fixed distribution from which to draw $n$ samples, and thus the extreme-value approach used to show Proposition \ref{prop:Noise} would not directly apply.} It also has a straightforward conceptual interpretation, implying that high-dimensional geometry amplifies the large-population bias towards data-rich individuals.

The next corollary considers any sequence of noise variances $(\sigma^2_{P,m})_{m=1}^\infty$ and $(\sigma^2_{R,m})_{m=1}^\infty$ satisfying (1) $\lim_{m \rightarrow \infty} \sigma_{R,m}^2 = c > 0$, i.e., the smaller noise variance $\sigma_{R,m}^2$ limits to a strictly positive constant; and (2) the ratio \[\frac{\sigma^2_{P,m}}{\sigma^2_{R,m}} \rightarrow \infty\] as $m\rightarrow \infty$, i.e., the noise variance $\sigma_{P,m}^2$ grows large relative to $\sigma_{R,m}^2$. Let $(i^*_{n,m},g^*_{n,m})$ denote the (random) best-match on an AI platform with $n$ data-rich and $n$ data-poor AI representations, whose noise variances are respectively given by $\sigma^2_{R,m}$ and $\sigma^2_{P,m}$. Corollary \ref{prop:LimitNoiseLargeVar} says that as this noise ratio grows large, the large-population probability of sampling from the data-rich group converges to one.

\begin{corollary}[Large Noise Disparity] \label{prop:LimitNoiseLargeVar} In the limit as the noise ratio and population size both grow large, the probability that Subject's match is from the data-rich group converges to 1, i.e.,
\[\lim_{m \rightarrow \infty} \left(\lim_{n\rightarrow \infty} P\left(g_{n,m}^* = R \right) \right)=1.\]
\end{corollary}

In each of these results, the systematic selection of data-rich individuals benefits the Subject: the platform selects matches whose AI representations are compatible with the Subject’s target, and data-rich individuals have more accurate representations.

But this selection process has  concerning distributional implications. Since the results are not specific to Subject's characteristics, they suggest that individuals who are better understood by AI systems will be systematically more in demand. This advantage could exacerbate existing inequalities if the property of being data-rich correlates with existing group identities. They moreover point to an emerging form of social stratification: in a world of AI representations, an individual's opportunities may depend not only on their inherent qualities but also on how effectively AI systems can capture and convey those qualities.

\section{Conclusion} \label{sec:Conclusion}

Many papers have compared human and AI evaluation on problems that are hard for both humans and machines (e.g., medical diagnosis). But there are important evaluations that are intrinsically easier for people---for example, no machine or expert knows better than ourselves whose company we enjoy. In such settings, human evaluation is  more accurate than AI evaluation, but also more costly. Can automated search over a sufficiently large number of candidates compensate for the AI's errors during evaluation? This paper suggests that  scale is valuable primarily when the evaluation problem is effectively low-dimensional. If match quality can be summarized by a small number of attributes, then large-scale AI-mediated search can indeed substantially improve outcomes. In contrast, when many hard-to-measure attributes matter, expanding algorithmic scale alone cannot substitute for direct human judgment.

Several important features of AI-mediated search are abstracted away from in the present model. First, Subject’s objective may be richer than Euclidean distance from a target outcome, and the platform may additionally face the problem of learning or inferring this objective from limited data. Second, both searchers and candidates may strategically manipulate or curate their AI representations. Incorporating these considerations would plausibly further limit the effectiveness of AI-mediated search, while also introducing new strategic and informational forces. They are left as interesting directions for future work.

\clearpage
\appendix

\section{Details on Figure \ref{fig:Simulate}} \label{app:Table}

Table \ref{tab:Data} reports the estimated values $\hat{d}^{\text{IP}}_k(2)$, $\hat{d}^{\text{AI}}_k(N)$ (for $N=10,000$), and $\frac{k}{k+1}$ depicted in Figure \ref{fig:Simulate}, as well as estimates for additional values of $k$. For small values of $k$,  $\hat{d}_k^{\text{AI}}(N)<\hat{d}^{\text{IP}}_k(2)$; indeed, for each of $k=1,5,10$, the expected distance to the AI best match is substantially lower than the expected distance to the in-person best match. This order reverses when the number of dimensions is large (in this case, 150 or larger). Both quantities approach $k/(k+1)$ from below as $k$ grows large.

\begin{table}[H]
\begin{center}
\begin{tabular}{cccc}
\toprule
$k$ & $\hat{d}^{\text{IP}}_k(2)$ & $\hat{d}^{\text{AI}}_k(N)$ & $\frac{k}{k+1}$ \\
\midrule
1   & 0.3346 & \cellcolor{gray!30}{0.0551} & 0.5 \\
5   & 0.7554 & \cellcolor{gray!30}{0.1743} & 0.8333 \\
10  & 0.8675 & \cellcolor{gray!30}{0.3664} & 0.9091 \\
50  & 0.9709 & \cellcolor{gray!30}{0.9126} & 0.9804 \\
100 & 0.9849 & \cellcolor{gray!30}{0.9818} & 0.9901 \\
125 & 0.9882 & \cellcolor{gray!30}{0.9871} & 0.9921 \\
150 & \cellcolor{gray!30}{0.9902} & 0.9903 & 0.9934 \\
175 & \cellcolor{gray!30}{0.9919} & 0.9920 & 0.9943 \\
200 & \cellcolor{gray!30}{0.9925} & 0.9931 & 0.9950 \\
225 & \cellcolor{gray!30}{0.9933} & 0.9941 & 0.9956 \\
250 & \cellcolor{gray!30}{0.9939} & 0.9950 & 0.9960 \\
275 & \cellcolor{gray!30}{0.9949} & 0.9957 & 0.9964 \\
300 & \cellcolor{gray!30}{0.9951} & 0.9959 & 0.9967 \\
400 & \cellcolor{gray!30}{0.9962} & 0.9972 & 0.9975 \\
500 & \cellcolor{gray!30}{0.9970} & 0.9979 & 0.9980 \\
750 & \cellcolor{gray!30}{0.9980} & 0.9986 & 0.9987 \\
1000 & \cellcolor{gray!30}{0.9985} & 0.9990 & 0.9990 \\
\bottomrule
\end{tabular}
\end{center}
\caption{This table reports estimates of $d^{\text{IP}}_k(2)$ and $d^{\text{AI}}_k(N)$  for $\sigma=0.05$ and $N=10,000$. The cell with the smallest value is highlighted in gray.} \label{tab:Data}
\end{table}

\section{Proofs of the Main Results}

Appendix \ref{supp:RealLife} and Appendix \ref{supp:AI}  present key supporting results about search in the in-person and AI representation regimes. These results are used in Appendix \ref{proof:Finite} to prove Proposition \ref{prop:Finite}, and in Appendix \ref{proof:2} to prove Theorem \ref{thm:2}. 

\subsection{Supporting Results about In-Person Search} \label{supp:RealLife}

Lemma \ref{lemm:RealLife} provides a closed-form expression for $d^{\text{IP}}_k(m)$, the expected distance between Subject's target and Subject's match under in-person search with $m$ samples. Lemma \ref{lemm:BoundRL} uses this lemma to bound the distance between $d^{\text{IP}}_k(2)$ and the benchmark $d^{\text{IP}}_k(1)=k/(k+1)$. 

\begin{lemma} \label{lemm:RealLife} For every number of dimensions $k$ and sample size $m$, the expected distance between Subject and Subject's match in the in-person regime is
\[
d^{\text{IP}}_k(m) = E\left[ \min_{1 \leq i \leq m} \| X_{i} \| \right]  = \frac{1}{k} B\left(\frac1k, m+1\right)
\]
where $B$ denotes the Beta function.
\end{lemma}

\begin{proof} First observe that
\begin{align*}
    P\left(\min_{1 \leq i \leq m} \| X_i \| > r\right) & =  P\left(\| X_1 \| > r\right)^m && \mbox{ since $X_i$ are i.i.d.} \\
    & = \left(1-\frac{V_{r,k}}{V_{1,k}}\right)^m = (1-r^k)^m
\end{align*}
 where $V_{r,k} = \frac{\pi^{k/2} r^k}{\Gamma\left(\frac{k}{2}+1\right)}$ is the volume of a ball with radius $r$. Thus
 \begin{align*}
    E\left[ \min_{1 \leq i \leq m} \| X_{i} \| \right] & = \int_0^1  P\left(\min_{1 \leq i \leq m} \| X_i \| > r\right) dr \\
    & = \int_0^1 (1-r^k)^m dr \\
    &= \frac{1}{k} B\left(\frac{1}{k},m+1\right)
 \end{align*}
 as desired.
 \end{proof}

\begin{lemma} \label{lemm:BoundRL} For every number of dimensions $k$,
\begin{equation} \label{identity:2}
d^{\text{IP}}_k(2) =  d^{\text{IP}}_k(1) - \frac{k}{(2k+1)(k+1)} = \frac{k}{k+1} - \frac{k}{(2k+1)(k+1)}.
\end{equation}
\end{lemma}

\begin{proof}
Recalling that Lemma \ref{lemm:RealLife} implies $d^{\text{IP}}_k(2) = \frac1k B\left(\frac1k,3\right)$, let us first show
\begin{equation} \label{identity:2}
\frac{1}{k} B\left(\frac{1}{k},3\right) = \frac{k}{k+1} - \frac{k}{(2k+1)(k+1)}.
\end{equation}
We can rewrite
\begin{align*}
    \frac1k B\left(\frac1k, 3\right) & = \frac1k \frac{\Gamma\left(\frac1k\right) \Gamma(3)}{\Gamma\left(\frac1k +3 \right)} && \mbox{ by definition of the Beta function}\\
    & = \frac1k\frac{2\Gamma\left(\frac1k\right) }{\Gamma\left(\frac1k +3 \right)} && \mbox{ since $\Gamma(3)=2!$} \\
    & = \frac1k \frac{2\Gamma\left(\frac1k\right)}{\left(\frac1k +2\right)\left(\frac1k +1\right)\frac1k \Gamma\left(\frac1k\right)} && \mbox{ by the Gamma function recurrence relation} \\
    & = \frac{2k^2}{(2k+1)(k+1)}
\end{align*}
Since
\[
\frac{k}{k+1}-\frac{k}{(2k+1)(k+1)}=\frac{2k^2}{(2k+1)(k+1)}
\]
we have the desired identity in (\ref{identity:2}). By similar arguments,
\begin{align*}
    \frac1k B\left(\frac1k, 2\right) & = \frac1k \frac{\Gamma\left(\frac1k\right) \Gamma(2)}{\Gamma\left(\frac1k +2 \right)}  = \frac1k\frac{\Gamma\left(\frac1k\right) }{\Gamma\left(\frac1k +2 \right)}   = \frac1k \frac{\Gamma\left(\frac1k\right)}{\left(\frac1k +1\right)\frac1k \Gamma\left(\frac1k\right)}  = \frac{k}{k+1}
\end{align*}
so 
$d^{\text{IP}}_k(2) = d^{\text{IP}}_k(1) - \frac{k}{(2k+1)(k+1)}$. That is, in expectation the second in-person draw reduces the distance between Subject and Subject's match by $\frac{k}{(2k+1)(k+1)}$.
\end{proof}

\subsection{Supporting Results about AI-Mediated Search} \label{supp:AI}

This section reports the following results: First, the expected distance between Subject's target and individual $i$ is increasing in the distance between their AI representations (Lemma \ref{lemm:CondExp}).\footnote{Lemma \ref{lemm:Monotone} uses this to show that Subject's expected payoff given $n$ samples is monotonically increasing in $n$.} Thus for every finite $n$, Subject's expected payoff can be upper bounded by Subject's payoff in an idealized setting where some AI representation perfectly matches the representation of the target, henceforth denoted $d^{\text{AI}}_k(\infty)$ (Corollary \ref{corr:Infinite}). Lemma \ref{lemm:Gaussian} proves that $d^{\text{AI}}_k(\infty)$ is equal to the expectation of the norm of a multivariate Gaussian vector, conditional on that norm being less than 1. Finally, Proposition \ref{lemm:AI} asymptotically bounds the difference between $d^{\text{AI}}_k(\infty)$ and the benchmark distance $k/(k+1)$, and Lemma \ref{lemm:1} shows that $d^{\text{AI}}_k(\infty)<k/(k+1)$ for every $k$.

\begin{lemma} \label{lemm:CondExp} 
$\mathbb{E}[\| X_i \| :\| Y_{i0} - Y_{0i} \| = s]$ 
is increasing in $s$.
\end{lemma}

\begin{proof}
Fix an arbitrary individual $i$ and define $Z_i \equiv \varepsilon_{i0} - \varepsilon_{0i} \sim \mathcal{N}(0,2\sigma^2I_k)$. Let
  $R \equiv \|X\|$ and $
  S \equiv \|X+Z\|$, where I drop the $i$ subscript throughout. I will show that the joint density $g_{R,S}$ satisfies the monotone likelihood ratio property, i.e., 
\begin{equation} \label{eq:MLRP}
  g_{R,S}(r,s)\,g_{R,S}(r',s') > g_{R,S}(r,s')\,g_{R,S}(r',s)
\end{equation}
for every $1\geq r> r' \geq 0$ and $s> s' \geq 0$.
It is well known that (\ref{eq:MLRP}) implies that $E(R \mid S=s)$ is increasing in $s$ (see e.g., \citet{Milgrom1981}).

Towards demonstrating (\ref{eq:MLRP}), I will first derive a closed-form expression for each of $g_R(r)$ and $g_{S \mid R=r}(s)$. Since $X$ is uniformly distributed on the unit ball, its norm $R$ has distribution function
\[G_R(r)=\mathbb{P}(R\le r) = \frac{V_{r,k}}{V_{1,k}}=r^k \quad \forall r\in [0,1]\]
 where $V_{r,k} = \frac{\pi^{k/2} r^k}{\Gamma\left(\frac{k}{2}+1\right)}$ is the volume of a ball with radius $r$. Differentiating yields
\begin{equation} \label{eq:fR}
  g_R(r) = k r^{k-1} \quad \forall r\in [0,1].
\end{equation}

Next turn to $g_{S|R=r}$. Conditional on $R=r$, we can write $X=r u$ for some unit vector $u$ on the $k$-dimensional unit sphere.
Hence
\[ S = \| X+\varepsilon \| = \| r u + \varepsilon \|\]
where the direction $u$ is irrelevant for the distribution of $S$. 

Define $T = \frac{S}{\sigma}$. Then $T \mid (R=r) = \| \frac{ru}{\sigma} + \frac{\varepsilon}{\sigma} \|$ has a noncentral chi distribution with dimension $k$ and mean vector $\mu = \frac{ru}{\sigma}$.
The pdf for such a distribution is 
\begin{equation} \label{eq:densityfT}
  h_T(t) = \frac{1}{\lambda^{\nu}} t^{\nu+1} \exp\left[-\frac{t^2 + \lambda^2}{2}\right] I_{\nu}(\lambda t)
\end{equation}
 where $\lambda = \|\mu\|=\frac{r}{\sigma}$, $\nu=\frac{k}{2}-1$,  and $I_{\nu}$ is the modified Bessel function of the first kind.
Applying the change of variable $S=\sigma T$ we have
\[g_{S \mid R=r}(s)  = \frac{1}{\sigma} h_T\!\left(\frac{s}{\sigma}\right)
\]
for all $s\geq 0$. Substituting from (\ref{eq:densityfT}) with $t=s/\sigma$, $\lambda=r/\sigma$, and $\nu=\frac{k}{2}-1$ yields
\begin{align}
  g_{S \mid R=r}(s) &=  \frac{1}{\sigma^{\nu+2}}
  \left(\frac{s^{\nu+1}}{r^\nu}\right)
  \exp\!\left[-\frac{r^2+s^2}{2\sigma^2}\right]
  I_{\nu}\!\left(\frac{rs}{\sigma^2}\right).
  \label{eq:gSgivenR}
\end{align}

Together with (\ref{eq:fR}), the joint density of $(R,S)$ is therefore
\begin{align}
  g_{R,S}(r,s)  &=  g_R(r)\, g_{S \mid R=r}(s) \nonumber \\
  &=  \frac{k}{\sigma^{k/2+1}} r^{k/2}  s^{k/2-1}
  \exp\!\left[-\frac{r^2+s^2}{2\sigma^2}\right]
  I_{\frac{k}{2}-1}\!\left(\frac{rs}{\sigma^2}\right).
  \label{eq:gRS}
\end{align}

A sufficient condition for (\ref{eq:MLRP}) is that $g_{R,S}$ is log-supermodular \citep{athey2002monotone}. Define
\begin{align}
  \phi(r,s)  &\equiv \ln g_{R,S}(r,s) \nonumber \\
  &=  \ln\frac{k}{\sigma^{k/2+1}} + \frac{k}{2}\ln r
  + \left(\frac{k}{2}-1\right)\ln s - \frac{r^2+s^2}{2\sigma^2} + \ln I_{\frac{k}{2}-1}\!\left(\frac{rs}{\sigma^2}\right).
  \label{eq:phi}
\end{align}
Since only the final term of (\ref{eq:phi}) has a nonzero cross-partial,
\[\frac{\partial^2}{\partial r \partial s}\phi(r,s) = \frac{\partial^2}{\partial r \partial s} \ln\left(I_{\frac{k}{2}}\left(\frac{rs}{\sigma^2}\right)\right)\]
Let $z=r s/\sigma^2$ and again denote $\nu=\frac{k}{2}-1$. Then
\[\frac{\partial}{\partial r}\ln I_{\nu}(z)= \frac{I_{\nu}'(z)}{I_{\nu}(z)}\cdot \frac{s}{\sigma^2}\]
  and so
  \begin{align*}
  \frac{\partial^2}{\partial r \partial s} \ln\left(I_{\nu}\left(\frac{rs}{\sigma^2}\right)\right) = \frac{\partial}{\partial s}\left[\frac{I_{\nu}'(z)}{I_{\nu}(z)}\cdot \frac{s}{\sigma^2}\right]
  & =\frac{\partial}{\partial s}\left[\frac{I_{\nu}'(z)}{I_{\nu}(z)}\right] \cdot \frac{s}{\sigma^2} + \frac{I_{\nu}'(z)}{I_{\nu}(z)}\cdot \frac{1}{\sigma^2} \\
  & = \frac{\partial}{\partial z}\left[\frac{I_{\nu}'(z)}{I_{\nu}(z)}\right] \cdot \frac{rs}{\sigma^4} + \frac{I_{\nu}'(z)}{I_{\nu}(z)}\cdot \frac{1}{\sigma^2}
\end{align*}
We can now sign the RHS using properties of $I_{\nu}$, the modified Bessel function of the first kind (see e.g., \citet{Baricz2010}). First, $I_{\nu}'(z)\geq 0$ and $I_{\nu}(z)\geq 0$ for all $z\geq 0$ (which is implied by $r,s\geq 0$); thus, $\frac{I_{\nu}'(z)}{I_{\nu}(z)}\geq 0 $. Moreover, for $\nu\ge-1/2$ (which is satisfied in our setting since $k\geq 1$), the function $I_{\nu}$ is log-convex on $(0,\infty)$, i.e., $\frac{d}{dz}\left[\frac{I_{\nu}'(z)}{I_{\nu}(z)}\right]\ge0$. Thus the entire RHS is positive, so $g(r,s)=\exp(\phi(r,s))$ is log-supermodular as desired. 
\end{proof}

\begin{corollary} \label{corr:Infinite} Define $d^{\text{AI}}_k(\infty) = E\left[ \| X_i \| : \|Y_{i0} - Y_{0i} \|=0\right].$
Then 
$d^{\text{AI}}_k(\infty) \leq d^{\text{AI}}_k(n)$ for all $ n \in \mathbb{Z}_+.$
\end{corollary}

\begin{proof}
Define $W=\min_{1\leq j \leq n} \| Y_{j0} - Y_{0j}\|$. Then  
\[d_k^{\text{AI}} (n) = E\left[\| X_i \| : \|Y_{i0}-Y_{0i} \| = W \right]\]
where clearly $P(W \geq 0)=1$. So
\[d_k^{\text{AI}} (n) \geq E\left[\| X_i \| : \|Y_{i0}-Y_{0i} \| = 0 \right]\]
follows by Lemma \ref{lemm:CondExp}.
\end{proof}

\begin{lemma} \label{lemm:Gaussian} For each $i$, $E\left[\| X_i \| : \|Y_{i0} - Y_{0i}\|=0 \right] = E[ \|Z \| : \| Z \| \leq 1 ]$ where $Z \sim \mathcal{N}(0, 2\sigma^2 I_k)$ and $I_k$ is the identity matrix in $k$ dimensions.
\end{lemma}

\begin{proof} Recall that $Y_{i0}=Y_{0i}$ is equivalent to 
$X_{i} + \varepsilon_{i0} = x_0 + \varepsilon_{0i}$, 
or more simply
\[X_i = Z_i\]
where $Z_i = \varepsilon_{0i}-\varepsilon_{i0} \sim \mathcal{N}(0,2\sigma^2I_k)$.
By Bayes’ rule, the posterior density of $X_i$ conditional on the event $\{X_i =  Z_i\}$ satisfies
\[
f_{X_i|Z_i=X_i}(x) =\frac{f_{X_i}(x)f_{Z_i}(x)}{\int_{x\in B_k} f_{X_i}(x)f_{Z_i}(x)dx}
\] Since $f_{X_i}(x)$ is constant on $B_k$ and zero elsewhere, this further simplifies to
\[
f_{X_i|Z_i=X_i}(x) =\frac{f_{Z_i}(x)}{\int_{x\in B_k}f_{Z_i}(x)dx}
\] 
That is, 
$
X_i \mid (Y_{i0}=Y_{0i}) \stackrel{d}{=} Z_i \mid \{\|Z_i\|\le 1\}$ which further implies
\[
E[\|X_i\| : Y_{i0}=Y_{0i}] = E[\|Z_i\| : \|Z_i\|\le 1]
\]
as desired.
\end{proof}

\begin{lemma} \label{lemm:AI} 
$d^{\text{AI}}_k(\infty) = \frac{k}{k+1} + o\left(\frac1k\right)$
\end{lemma}

\begin{proof}
Let $Z\sim N(0,2\sigma^2 I_k)$ and $R=\|Z\|$. 
Since $W=\frac{R^2}{2\sigma^2}\sim\chi^2_k$ has density
\[
f_W(w)=\frac{1}{2^{k/2}\Gamma(\frac{k}{2})}w^{\frac{k}{2}-1}e^{-w}\]
the change of variables
$r=\sqrt{2\sigma^2 w}$ yields
\begin{equation}\label{eq:fr}
f_R(r) = \frac{2^{1-k}}{\Gamma(\frac{k}{2}) \sigma^k} r^{k-1} \exp\!\left(-\frac{r^2}{4\sigma^2}\right) \quad \forall r\geq 0.
\end{equation}
Therefore
\[\mathbb{E}[ R \mid R \leq 1 ]
= \frac{\int_0^1 r f_R(r)\,dr}{\int_0^1 f_R(r)\,dr}
= \frac{\int_0^1 r^{k} \exp\!\left(-\frac{r^2}{4\sigma^2}\right) dr}{\int_0^1 r^{k-1} \exp\!\left(-\frac{r^2}{4\sigma^2}\right) dr}.\]
Define
\[I_1(k) := \int_0^1 r^{k} \exp\!\left(-\frac{r^{2}}{4\sigma^{2}}\right)dr,
\qquad
I_2(k) := \int_0^1 r^{k-1} \exp\!\left(-\frac{r^{2}}{4\sigma^{2}}\right)dr.\]
We will show
\begin{equation}\label{eq:AsymptoticRatio}
\frac{I_1(k)}{I_2(k)}=\frac{k}{k+1}+O\!\left(\frac{1}{k^2}\right)
\end{equation}

Define $g(r):=\exp\!\left(-\frac{r^2}{4\sigma^2}\right)$ on the domain $[0,1]$. Note that $g\in C^2([0,1])$ and
\[
g(1)=e^{-\frac{1}{4\sigma^2}},\qquad
g'(r)=-\frac{r}{2\sigma^2}e^{-\frac{r^2}{4\sigma^2}},\qquad
g'(1)=-\frac{1}{2\sigma^2}e^{-\frac{1}{4\sigma^2}}.
\]
Moreover, $g''$ is continuous on $[0,1]$, so $M:=\sup_{r\in[0,1]}|g''(r)|<\infty$.
For each integer $k\ge 1$, integrate by parts with $u=g(r)$ and $dv=r^k\,dr$:
\begin{align*}
\int_0^1 r^k g(r)\,dr & =\left[\frac{r^{k+1}}{k+1}g(r)\right]_0^1-\frac{1}{k+1}\int_0^1 r^{k+1}g'(r)\,dr \\
& =\frac{g(1)}{k+1}-\frac{1}{k+1}\int_0^1 r^{k+1}g'(r)\,dr.
\end{align*}
Apply integration by parts again to the remaining integral, with $u=g'(r)$ and $dv=r^{k+1}\,dr$:
\begin{align*}
\int_0^1 r^{k+1}g'(r)\,dr & =\left[\frac{r^{k+2}}{k+2}g'(r)\right]_0^1-\frac{1}{k+2}\int_0^1 r^{k+2}g''(r)\,dr \\
& =\frac{g'(1)}{k+2}-\frac{1}{k+2}\int_0^1 r^{k+2}g''(r)\,dr.
\end{align*}
Combining these expressions yields
\[\int_0^1 r^k g(r)\,dr
=\frac{g(1)}{k+1}-\frac{g'(1)}{(k+1)(k+2)}+\frac{1}{(k+1)(k+2)}\int_0^1 r^{k+2}g''(r)\,dr.\]
The remainder is bounded as
\[\left|\int_0^1 r^{k+2}g''(r)\,dr\right|\le M\int_0^1 r^{k+2}\,dr=\frac{M}{k+3},\]
so
\begin{equation}\label{eq:Ikexpansion}
\int_0^1 r^k g(r)\,dr
=\frac{g(1)}{k+1}-\frac{g'(1)}{(k+1)(k+2)}+O\!\left(\frac{1}{k^3}\right).
\end{equation}

Thus
\begin{align*}
I_1(k)&=\frac{g(1)}{k+1}-\frac{g'(1)}{(k+1)(k+2)}+O\!\left(\frac{1}{k^3}\right) \\
I_2(k)&=\frac{g(1)}{k}-\frac{g'(1)}{k(k+1)}+O\!\left(\frac{1}{k^3}\right)
\end{align*}
Factor out $g(1)=e^{-1/(4\sigma^2)}$ and define $A:= -\frac{g'(1)}{g(1)}=\frac{1}{2\sigma^2}$. Then
\begin{align*}
I_1(k)&=g(1)\left(\frac{1}{k+1}+\frac{A}{(k+1)(k+2)}+O\!\left(\frac{1}{k^3}\right)\right) \\
I_2(k)&=g(1)\left(\frac{1}{k}+\frac{A}{k(k+1)}+O\!\left(\frac{1}{k^3}\right)\right)
\end{align*}
Their ratio simplifies to
\[\frac{I_1(k)}{I_2(k)}
=\frac{k+\frac{Ak}{k+2}+O(\frac{1}{k})}{(k+1)+A+O(\frac{1}{k})}.\]
Since $\frac{Ak}{k+2}=A+O(1/k)$, the numerator is $k+A+O(1/k)$ and the denominator is $k+1+A+O(1/k)$, hence
\[\frac{I_1(k)}{I_2(k)}
=\frac{k+A+O(\frac{1}{k})}{k+1+A+O(\frac{1}{k})}
=\frac{k}{k+1}+O\!\left(\frac{1}{k^2}\right).\]
This establishes \eqref{eq:AsymptoticRatio}, and therefore
\[\mathbb{E}[R\mid R\le 1]=\frac{I_1(k)}{I_2(k)}=\frac{k}{k+1}+o\!\left(\frac{1}{k}\right),\]
as desired.
\end{proof}

\begin{lemma} \label{lemm:1} For every positive integer $k$,
$d^{\text{AI}}_k(\infty) < \frac{k}{k+1}.$
\end{lemma}

\begin{proof}
Let $T$ be the random variable with density function \[g_T(r) = \frac{r^{k-1} \exp\left(-\frac{r^2}{4\sigma^2}\right)}{\int_0^1 r^{k-1} \exp\left(-\frac{r^2}{4\sigma^2}\right) dr}\]
on $[0,1]$, observing that 
\[E[T] = d_k^{\text{AI}}(\infty) = \frac{\int_0^1 r^{k} \exp\left(-\frac{r^2}{4\sigma^2}\right) dr}{\int_0^1 r^{k-1} \exp\left(-\frac{r^2}{4\sigma^2}\right) dr}.\]
Consider also a Beta$(k,1)$ random variable $B$, whose density function on $[0,1]$ is
$
g_{B}(r) = k r^{k-1}$ and expectation is 
$
E[B] = \frac{k}{k+1}.
$ The likelihood ratio 
\[
\frac{g_{T}(r)}{g_{B}(r)} = \frac{1}{kr^{k-1}} \frac{r^{k-1}\exp\left(-\frac{r^{2}}{4\sigma^{2}}\right)}{\int_{0}^{1} r^{k-1}\exp(-r^{2}/4\sigma^{2})dr}
= 
\frac{1}{k}\frac{\exp(-r^{2}/4\sigma^{2})}
{ \int_{0}^{1} r^{k-1}\exp(-r^{2}/4\sigma^{2})dr}.
\]
is strictly decreasing in $r$, implying that $\left\{\frac{g_{T}(r)}{g_{B}(r)}\right\}_{r\geq 0}$ has the monotone likelihood ratio property. Thus the distribution of $B$ first-order stochastically dominates the distribution of $T$, implying in particular that $E[T]<E[B]$, or equivalently,
\[\frac{\int_0^1 r^{k} \exp\left(-\frac{r^2}{4\sigma^2}\right) dr}{\int_0^1 r^{k-1} \exp\left(-\frac{r^2}{4\sigma^2}\right) dr} < \frac{k}{k+1}\]
as desired.\end{proof}

\subsection{Proof of Proposition \ref{prop:Finite}} \label{proof:Finite}

By Lemma \ref{lemm:RealLife}, 
\[d^{\text{IP}}_k(m) = \frac{1}{k} B\left(\frac1k, m+1\right)\]
For every fixed $k$, this expression converges to zero as $m$ grows large. Since the quantity
$\mathbb{E}\left[ \| X_i \| : \|Y_{i0} - Y_{0i} \|=0\right]$ is strictly positive, 
we can identify a finite $m$ such that
\[d^{\text{IP}}_k(m) <\mathbb{E}\left[ \| X_i \| : \|Y_{i0} - Y_{0i} \|=0\right].\]
Then by Corollary \ref{corr:Infinite}, 
\[d^{\text{IP}}_k(m)  < d^{\text{AI}}_k(n) \quad \forall n \in \mathbb{Z}_+\]
implying that the AI equivalent sample size exists and is finite.

\subsection{Proof of Theorem \ref{thm:2}} \label{proof:2}

By Lemma \ref{lemm:BoundRL},
\[d^{\text{IP}}_k(2) = \frac{k}{k+1} - \frac{k}{(2k+1)(k+1)}\]
where $\frac{k}{(2k+1)(k+1)} = \Theta\left(\frac1k\right)$. That is, $d^{\text{IP}}_k(2)$ is smaller than $ \frac{k}{k+1}$ by a quantity that is on the order of $1/k$. By Lemma \ref{lemm:AI},
\[d^{\text{AI}}_k(\infty) = \frac{k}{k+1} + o\left(\frac{1}{k}\right)\]
i.e., the difference between $d^{\text{AI}}_k(\infty) $ and $k/(k+1)$ vanishes faster than $1/k$ as $k$ grows large. 
Moreover by Corollary \ref{corr:Infinite}, $d^{\text{AI}}_k(\infty)$ is a uniform lower bound on $d^{\text{AI}}_k(n)$ for all $n$. Thus when $k$ is sufficiently large, we have
\[d^{\text{IP}}_k(2) < d^{\text{AI}}_k(n)\]
for all $n \in \mathbb{Z}_+$, so $m_k^* \leq 2$. Since also 
\[d^{\text{IP}}_k(1) = \mathbb{E}[\|X_1\|] = d^{\text{AI}}_k(1)\]
we have $m_k^* >1$. Thus $m_k^*=2$ is the AI equivalent sample size for $k$ sufficiently large.

\section{Proofs of Results in Section \ref{sec:Supplementary}}

\subsection{Proof of Proposition \ref{prop:Noise}}

For each individual $i$ in group $g$, define
\[S^g_i = \|Y^g_{i0}-Y^g_{0i}\| = \|X_i^g + Z^g_i\| \quad \mbox{where} \quad Z^g_i  = \varepsilon^g_{i0}-\varepsilon^g_{0i}\]
to be the distance between this individual's AI clone $Y^g_{i0}$ and Subject's AI clone $Y^g_{0i}$ in their interaction.  
Recall that $X_i^g$ is drawn from a uniform distribution on the unit ball while $Z^g_i \sim \mathcal{N}(0,\nu_g I_k)$, where $\nu_R \equiv \sigma_R^2+\sigma_0^2$ and $\nu_P \equiv \sigma_P^2 + \sigma_0^2$. So the random variable 
\[Y^g_{i0}-Y^g_{0i} = X_i^g + Z_i^g\] admits a density $f_k^g$ that is locally Lipschitz around the origin. That is, there is an $L>0$ such that
\[\vert f_k^g(x) - f_k^g(0) \vert \leq L \| x\|\]
for all $x$ sufficiently close to the origin. Integrating this Lipschitz bound over the ball of small radius $s$, we obtain
\begin{align*}
    \left\vert \int_{\| x \| \leq s} (f_k^g(x) - f_k^g(0)) dx \right\vert & \leq \int_{\| x \| \leq s} \left\vert f_k^g(x) - f_k^g(0) \right\vert dx \\
    & \leq \int_{\| x \| \leq s} L \| x \| dx = O(s^{k+1})
\end{align*}
Thus 
\[\int_{\| x \| \leq s} f_k^g(x) dx = \int_{\| x \| \leq s} f_k^g(0) dx + O(s^{k+1})\]
or equivalently,
\begin{equation} \label{eq:Fg}
H_g(s) = f_k^g(0) V_{1,k}s^k + O(s^{k+1})
\end{equation}
where $H_g$ denotes the CDF of $S_i^g$, and $V_{1,k}=\frac{\pi^{k/2}}{\Gamma(\frac{k}{2}+1)}$ is the volume of the unit ball. 

Now let $M_n^g = \min_{1\leq i \leq n} S_i^g$ be the smallest distance to Subject among individuals in group $g$. Since $S_g^i$ are iid,
\[P\left(M_n^g > s\right) = (1-H_g(s))^n,\]
and plugging in (\ref{eq:Fg}) yields
\[P\left(M_n^g > s\right)= \left(1-f_k^g(0) V_{1,k} s^k + O(s^{k+1})\right)^n.\]
Consider an arbitrary $r > 0$ and set $s_n = \left(\frac{r}{n}\right)^{1/k}$. Then as $n$ grows large,
\[P\left(M_n^g > s_n\right) = \left(1-f_k^g(0) V_{1,k}\left(\frac{r}{n}\right) + O\left( \left(\frac{1}{n}\right)^{(k+1)/k}\right)\right)^n \longrightarrow e^{-\lambda_g r}\]
where $\lambda_g \equiv f_k^g(0)V_{1,k}$.
If we define the scaled random variable $Y_n^g = n (M_n^g)^k$, then equivalently
\[P(Y_n^g > r) \longrightarrow e^{-\lambda_g r} \quad \mbox{as $n$ grows large}\]
So $Y_n^g \rightarrow Y^g$ in distribution, where the limiting random variable $Y^g$ follows an exponential distribution with rate parameter $\lambda_g$. Since moreover $(Y_n^R,Y_n^P)$ are independent for each $n$, and the limiting random variables $(Y^R, Y^P)$ are independent, 
\begin{align*}
    P(Y_n^R < Y_n^P) \longrightarrow  P(Y^R < Y^P) = \frac{\lambda_R}{\lambda_R + \lambda_P}
\end{align*}
where the final equality is a well-known comparison for the hit rate of two independent exponential random variables. Since the latter expression reduces to $\frac{f_k^R(0)}{f_k^R(0) + f_k^P(0)}$, and the event $\{Y_n^R < Y_n^P\}$ is identical to the event $\{M_n^R < M_n^P\}$, we have the first part of the desired result.

It remains to show that $f_k^R(0)>f_k^P(0)$. Using the  parametric distributions of $X_i^g$ and $Z_g^i$, we can analytically derive the density at zero to be
\begin{equation} \label{eq:pg}
f_k^g(0) = \int_{x \in B_k} \frac{1}{V_{1,k}} \frac{1}{(2\pi\nu_g)^{k/2}} \exp\left(-\frac{\|x\|^2}{2\nu_g}\right)dx >0
\end{equation}
Moving to spherical coordinates, 
\[\int_{x \in B_k} \exp\left(-\frac{\|x\|^2}{2\nu_g}\right)dx
= |S^{k-1}| \int_{0}^{1} r^{k-1} \exp\left(-\frac{r^2}{2\nu_g}\right)dr\]
where
$|S^{k-1}|=\frac{2 \pi^{k/2}}{\Gamma(\frac{k}{2})}$ is the surface area of the unit sphere. 
Further performing the change of variable $t = \frac{r^2}{2 \nu_g}$ yields
\begin{align*}\int_{x \in B_k} \exp\left(-\frac{\|x\|^2}{2 \nu_g}\right)dx & = |S^{k-1}| \frac{(2\nu_g)^{k/2}}{2} \int_{0}^{\frac{1}{2\nu_g}} t^{\frac{k}{2}-1}e^{-t}dt 
\end{align*}
Plugging this into (\ref{eq:pg}) and noting that $\Gamma\left(\frac{k}{2}+1\right)=\frac{k}{2}\Gamma\left(\frac{k}{2}\right)$, we obtain
\begin{equation} \label{eq:fkg0}
f_k^g(0) = \frac{k}{2 \pi^{k/2}} \int_0^{\frac{1}{2\nu_g}} t^{\frac{k}{2}-1} e^{-t}dt
\end{equation}
which is clearly monotonically decreasing in $\nu_g$. Thus $\nu_R < \nu_P$ implies $f_k^R(0) > f_k^P(0)$, as desired.

\subsection{Proof of Corollary \ref{prop:LimitNoiseLargek}}

From Proposition \ref{prop:Noise}, we know that the limiting probability as the population grows large is
\[\lim_{n\rightarrow \infty} P(g_{n,k}^*=R) = \frac{f^R_k(0)}{f^R_k(0) + f^P_k(0)}. \]
Moreover by (\ref{eq:fkg0}), 
\[\frac{f_k^R(0)}{f_k^P(0)} = \frac{\int_0^{\frac{1}{2\nu_R}} t^{\frac{k}{2}-1}e^{-t}dt}{\int_0^{\frac{1}{2\nu_P}} t^{\frac{k}{2}-1}e^{-t}dt}\]
Observe that
\[e^{-\frac{1}{2\nu_g}} \frac{2}{k}\left(\frac{1}{2\nu_g}\right)^{\frac{k}{2}} \leq \int_0^{\frac{1}{2\nu_g}} t^{\frac{k}{2}-1}e^{-t}dt  \leq \frac{2}{k}\left(\frac{1}{2\nu_g}\right)^{\frac{k}{2}}  \]
since $e^{-\frac{1}{2\nu_g}} \leq e^{-t} \leq 1$ on the domain $t\in [0,1/(2\nu_g)]$. Thus
\[\frac{\int_0^{\frac{1}{2\nu_R}} t^{\frac{k}{2}-1}e^{-t}dt}{\int_0^{\frac{1}{2\nu_P}} t^{\frac{k}{2}-1}e^{-t}dt} \geq \frac{e^{-\frac{1}{2\nu_R}} \frac{2}{k}\left(\frac{1}{2\nu_R}\right)^{\frac{k}{2}}}{\frac{2}{k}\left(\frac{1}{2\nu_P}\right)^{\frac{k}{2}}} = e^{-\frac{1}{2\nu_R}} \left(\frac{\nu_P}{\nu_R}\right)^{\frac{k}{2}}\]
which converges to $\infty$ as $k$ grows large, since $\nu_P/\nu_R = (\sigma_0^2+\sigma_P^2)/(\sigma_0^2 + \sigma_R^2)>1$ by assumption that $\sigma_R^2 < \sigma_P^2$.

\subsection{Proof of Corollary \ref{prop:LimitNoiseLargeVar}}

By (\ref{eq:pg}),
\[\frac{f^R_k(0)}{f^P_k(0)} =\left(\frac{\sigma_{0,m}^2+\sigma_{P,m}^2}{\sigma_{0,m}^2+\sigma_{R,m}^2}\right)^{k/2} \cdot \frac{\int_{B_k}\exp\!\left(-\frac{\|x\|^2}{2(\sigma_{0,m}^2+\sigma_{R,m}^2)}\right)\,dx}{\int_{B_k}\exp\!\left(-\frac{\|x\|^2}{2(\sigma_{0,m}^2+\sigma_{P,m}^2)}\right)\,dx}.
\]
By assumption that $\sigma_{P,m}^2/\sigma_{R,m}^2 \rightarrow \infty$ as $m \to \infty$ while $\sigma_{0,m}^2$ is fixed, we have
\[\left(\frac{\sigma_{0,m}^2+\sigma_{P,m}^2}{\sigma_{0,m}^2+\sigma_{R,m}^2}\right)^{k/2} 
\longrightarrow \infty\]
as $m \rightarrow \infty$. Next let
\[D_m \equiv \int_{B_k} \exp\!\left(-\frac{\|x\|^2}{2(\sigma_{0,m}^2+\sigma_{P,m}^2)}\right)\,dx
\]
be the integral in the denominator.  Since $\sigma_{P,m}^2 \to \infty$ and $\sigma_{0,m}^2 \ge 0$, we have
$\sigma_{0,m}^2+\sigma_{P,m}^2 \to \infty$. Hence, for each fixed $x \in B_k$,
\[
\exp\!\left(-\frac{\|x\|^2}{2(\sigma_{0,m}^2+\sigma_{P,m}^2)}\right)
\longrightarrow 1
\]
as $m \rightarrow \infty$. Moreover, for all $m$ and all $x \in B_k$,
\[ 0 \le \exp\!\left(-\frac{\|x\|^2}{2(\sigma_{0,m}^2+\sigma_{P,m}^2)}\right) \le 1 ,
\]
and the constant function $1$ is integrable over the bounded set $B_k$. Thus by the Dominated Convergence Theorem,
\[D_m \longrightarrow \int_{B_k} 1 \, dx = V_{1,k},\]
where $V_{1,k}$ denotes the volume of the unit ball in $\mathbb{R}^k$. In particular, $D_m$ is bounded above and bounded away from zero for all sufficiently large $m$.

Now define
\[
N_m \equiv \int_{B_k}
\exp\!\left(-\frac{\|x\|^2}{2(\sigma_{0,m}^2+\sigma_{R,m}^2)}\right)\,dx 
\]
to be the integral in the numerator. By assumption,
\[\lim_{m \to \infty} N_m = c > 0 ,\]
and hence $\liminf_{m \to \infty} N_m \ge c > 0$.

Combining these observations,
\[\liminf_{m \to \infty} \frac{N_m}{D_m} = \frac{\liminf_{m \to \infty} N_m}{\lim_{m \to \infty} D_m} \ge \frac{c}{V_{1,k}} > 0 .
\]
Therefore,
\[ \frac{f^R_k(0)}{f^P_k(0)} = \left(\frac{\sigma_{0,m}^2+\sigma_{P,m}^2}{\sigma_{0,m}^2+\sigma_{R,m}^2}\right)^{k/2}  \cdot \frac{N_m}{D_m}
\longrightarrow \infty\]
as $m \rightarrow \infty$.  The desired result then follows from Proposition~\ref{prop:Noise}.

\section{Additional Results}

Section \ref{app:Monotonen} proves a result that says that Subject's expected payoff in the AI representation regime is monotonically increasing in $n$, the number of candidates. Section \ref{app:Noise} proves a result that says that the expected distance between Subject and Subject's AI best match (in the infinite sample benchmark) converges to zero as $\sigma \rightarrow 0$. 

\subsection{Monotonicity in $n$} \label{app:Monotonen}

\begin{lemma} \label{lemm:Monotone} Subject's expected payoff in the AI representation regime is monotonically increasing in $n$; that is, $d^{\text{AI}}_k(n+1) < d^{\text{AI}}_k(n)$ for every $n \in \mathbb{Z}_+$.
\end{lemma}

\begin{proof}
For each $i$ define
\[S_i \equiv \| Y_{i0} - Y_{0i} \| = \|X_i +\varepsilon_{i0} - \varepsilon_{0i} \| = \|X_i + Z_i \| \]
where $Z_i \sim \mathcal{N}(0,2\sigma^2I_k)$. Consider a single probability space on which the infinite sequence $\{(X_i,S_i)\}_{i=1}^\infty$ is defined, noting that this tuple is independent across $i$. On this space, the random variable
\[i^*_n = \argmin_{1\leq i\leq n} S_i\]
is well-defined for every $n$. Moreover we can write
\begin{align*}
    d^{AI}_k(n+1) & = \mathbb{E}\left[\| X_{n+1} \| \cdot \mathbbm{1}( S_{n+1} < S_{i^*_n}) + \|X_{i^*_n}\| \cdot \mathbbm{1}( S_{n+1} \geq S_{i^*_n})\right] \\
    d^{\text{AI}}_k(n) & = \mathbb{E}\left[\|X_{i^*_n}\| \cdot \mathbbm{1}( S_{n+1} < S_{i^*_n}) + \|X_{i^*_n}\| \cdot\mathbbm{1}( S_{n+1} \geq S_{i^*_n})\right]
\end{align*}
so (by linearity of expectation) it is sufficient to show 
\[\mathbb{E}\left[\|X_{n+1}\| \cdot \mathbbm{1}( S_{n+1} < S_{i^*_n})\right] \leq \mathbb{E}\left[\|X_{i^*_n}\|  \cdot \mathbbm{1}( S_{n+1} < S_{i^*_n}) \right]\]
or equivalently 
\[\mathbb{E}\left[\| X_{n+1} \| : S_{n+1} < S_{i^*_n} \right] \leq \mathbb{E}\left[ \| X_{i^*_n} \| : S_{n+1} < S_{i^*_n} \right] \]

Conditional on $S_{n+1}$, the variable $X_{n+1}$ is independent of the sequence $\{(X_i,S_i)\}_{i=1}^n$. Since $S_{i^*_n}$ is a measurable function of $\{(X_i,S_i)\}_{i=1}^n$, also $X_{n+1} \indep S_{i^*_n} \mid S_{n+1}$. Thus 
\begin{equation} \label{eq:m1}
E\left[\| X_{n+1}\| : S_{n+1},S_{i_n^*}\right] = \mathbb{E}\left[\| X_{n+1} \| : S_{n+1}\right] =: m(S_{n+1})
\end{equation}
So also
\begin{align*}
    E\left[\| X_{n+1} \| : S_{n+1} < S_{i_n^*}\right] & =  E\left[ E\left[\| X_{n+1} \| :  S_{n+1},S_{i_n^*}\right] : S_{n+1} < S_{i_n^*}\right] && \mbox{ by L.I.E.}\\
    & = E\left[m(S_{n+1}) \mid S_{n+1} < S_{i^*_n}\right] && \mbox{ by (\ref{eq:m1})}
\end{align*}
Similarly,
\begin{equation} \label{eq:m2}
E\left[\| X_{i_n^*} \| : S_{n+1},S_{i_n^*}\right] = \mathbb{E}\left[\| X_{i_n^*} \| : S_{i_n^*}\right] =: m(S_{i_n^*})
\end{equation}
and
\begin{align*}
E\left[ \| X_{i_n^*} \| : S_{n+1} < S_{i_n^*}\right] & =  E\left[ E\left[X_{i_n^*} \mid S_{n+1},S_{i_n^*}\right] \mid S_{n+1} < S_{i_n^*}\right] \\
& = E\left[m(S_{i_n^*}) \mid S_{n+1} < S_{i^*_n}\right]
\end{align*}
Finally, observe that on the event $\{S_{n+1} < S_{i^*_n}\}$ we have $S_{n+1}(\omega) < S_{i^*_n}(\omega)$ pointwise for every $\omega \in \{S_{n+1} < S_{i^*_n}\}$. Since by Lemma \ref{lemm:CondExp} the function $m(\cdot)$ is monotonically increasing,
\[m(S_{n+1}(\omega)) \leq m(S_{i_n^*}(\omega))\]
also holds pointwise on $\{S_{n+1} < S_{i^*_n}\}$. This inequality is thus preserved by taking an expectation on $\{S_{n+1} < S_{i^*_n}\}$, i.e.,
\[E\left[m(S_{n+1}) \mid S_{n+1} < S_{i^*_n}\right] \leq E\left[m(S_{i_n^*}(\omega) \mid S_{n+1} < S_{i^*_n}\right].\]
Finally, by (\ref{eq:m1}) and (\ref{eq:m2}), this is equivalent to the desired statement
  \[ E\left[\| X_{n+1} \| : S_{n+1} < S_{i_n^*}\right] \leq  E\left[\| X_{i_n^*} \| : S_{n+1} < S_{i_n^*}\right].\]
  \end{proof}

\subsection{Limit as $\sigma\rightarrow 0$}
\label{app:Noise}
Define
\[
I_1(\sigma) := \int_0^1 r^{k} \exp\left(-\frac{r^{2}}{4\sigma^{2}}\right)dr
\quad\text{and}\quad
I_2(\sigma) := \int_0^1 r^{k-1} \exp\left(-\frac{r^{2}}{4\sigma^{2}}\right)dr.
\]
Apply the change of variable 
$t= \frac{r^2}{4 \sigma^2}$ to obtain
\[I_1(\sigma)= 2^{k}\sigma^{k+1}
\int_{0}^{1/(4\sigma^2)} t^{\frac{k-1}{2}}e^{-t}dt = 2^{k}\sigma^{k+1} \gamma \left(\frac{k+1}{2},\frac{1}{4\sigma^2}\right),
\]
where 
$
\gamma(s,x) = \int_{0}^{x} t^{s-1}e^{-t}dt$ is the lower incomplete gamma function.
The same substitution for $I_2(\sigma)$ yields
\[
I_2(\sigma) = 2^{k-1}\sigma^{k} \int_{0}^{1/(4\sigma^2)} t^{\frac{k}{2}-1}e^{-t}dt= 2^{k-1}\sigma^{k} \gamma\left(\frac{k}{2},\frac{1}{4\sigma^2}\right).
\]
From the above identities, we have
\[
\frac{I_1(\sigma)}{I_2(\sigma)} = 2\sigma \frac{\gamma\left(\frac{k+1}{2},\frac{1}{4\sigma^2}\right)}{\gamma\left(\frac{k}{2},\frac{1}{4\sigma^2}\right)}.
\]
Since for each fixed $s>0$, $
\lim_{x\to \infty} \gamma(s,x) = \Gamma(s)$ (where $\Gamma$ represents the Gamma function),
\[
\lim_{\sigma \rightarrow 0} \gamma\left(\frac{k+1}{2},\frac{1}{4\sigma^2}\right) = \Gamma\left(\frac{k+1}{2}\right) \quad \mbox{ and } \lim_{\sigma \rightarrow 0} \gamma\left(\frac{k}{2},\frac{1}{4\sigma^2}\right) = \Gamma\left(\frac{k}{2}\right) 
\]
Thus $\lim_{\sigma \rightarrow 0} \frac{I_1(\sigma)}{I_2(\sigma)} = 0 \cdot \frac{\Gamma \left(\frac{k+1}{2}\right)}{\Gamma\left(\frac{k}{2}\right)}= 0$, 
as desired.

\subsection{Sequential Search} \label{app:Sequential}

Time $t = 0, 1, 2, \dots$ is discrete. In each period $t\geq1$, an individual $X_t$ is drawn from a uniform distribution on the unit ball $B_k$, and error terms $\eps_{t0}$ and $\eps_{0t}$ are drawn from $\mathcal{N}(0,\sigma^2 I_k)$. All random variables are mutually independent. The state space is
$
\Omega = (B_k)^\infty \times \mathbb{R}^\infty \times \mathbb{R}^\infty$ with typical element
$
\omega = \left(\left(X_t\right)_{t=1}^\infty,\left(\eps_{t0}\right)_{t=1}^\infty,\left(\eps_{0t}\right)_{t=1}^\infty\right).$

At time $t=0$, Subject chooses between search in the in-person regime and the AI representation regime. If Subject chooses the in-person regime, then in each subsequent period $t$, Subject observes $X_t$ (corresponding to meeting individual $t$) and decides whether to stop and match with the closest individual so far or to continue searching. At the time of stopping $\tau$, Subject receives the payoff
\[
- \min\{\|X_1\|,\|X_2\|,\ldots,\|X_{\tau}\|\} - c_{\text{IP}}(\tau),
\]
where the in-person search cost $c_{\text{IP}}(\cdot)$ is an increasing function of $\tau$. 

If instead Subject chooses the AI representation regime, then Subject first pays a fixed cost $\kappa$ to use the AI-based platform. In each period $t$,  Subject observes the representation match quality
\[
S_t = \| X_t + \eps_{t0} - \eps_{0t}\|
\]
and decides whether to stop or continue. If Subject stops at time $\tau$, the realized payoff is
\[
 - \| X_{i^*(S_{\leq \tau})}\| - c_{\text{AI}}(\tau) - \kappa,
\]
where $S_{\leq t} = (S_1, \dots, S_t)$ is Subject's history at time $t$, and
$
i^*\left(S_{\leq \tau}\right) = 
\arg \min_{1 \leq i \leq \tau} S_i
$
denotes the individual prior to time $\tau$ whose AI representation is closest to Subject's AI representation. The AI platform search cost $c_{\text{AI}}(\cdot)$ is an increasing function of $\tau$. (It is reasonable to assume that $c_{\text{AI}}(\tau) \leq c_{\text{IP}}(\tau)$ for all $\tau$,  but this is not necessary for the results.)

Subject's strategy thus consists of an action $a \in \{\text{IP}, \text{AI}\}$ and a stopping rule $\tau_a$, which is a map $\Omega \to [0,\infty)$ that is measurable with respect to the observed histories in the chosen regime. Formally,  let $\mathcal{F}_t^{\text{IP}}$ denote the filtration generated by the in-person history $X_{\leq t} = (X_1,\dots,X_t)$, and let $\mathcal{F}_t^{\text{AI}}$ denote the filtration generated by the  AI history $S_{\le t}$.  Then if $a=\text{IP}$, the stopping rule must satisfy 
$\{\tau_a \leq t\} \in \mathcal{F}_t^{\text{IP}}$ for every $t$,  and  if $a=\text{AI}$, the stopping rule must satisfy
$\{\tau_a \leq t\} \in \mathcal{F}_t^{\text{AI}}$ for every $t$.

The following result says that when the number of dimensions is sufficiently large, and if the fixed cost of joining the AI platform exceeds the cost of searching over two individuals in person, then Subject optimally chooses the in-person regime.  

\begin{proposition} Suppose
$\kappa > c_{\mathrm{IP}}(2)$. Then for all $k$ sufficiently large, Subject optimally chooses $a=\text{IP}$ at period $t=0$.
\end{proposition}

\begin{proof} Since $(X_t,S_t)$ are i.i.d. across $t$, 
\[E\left[\|X_{i^*(S_{\leq t})}\| : S_{i^*(S_{\leq t})}\right] = E\left[\|X_{i^*(S_{\leq t})}\| : S_{\leq t}\right].\]
We have already shown in Lemma \ref{lemm:CondExp} that $E\left[\| X_{i^*(S_{\leq t})} \| \mid S_{i^*(S_{\leq t})}=s\right]$ is increasing in $s$, and thus
\[E\left[\|X_{i^*(S_{\leq t})}\| : S_{i^*(S_{\leq t})}=s\right] \geq d_k^{\text{AI}}(\infty) =  E\left[\| X_{i^*(S_{\leq t})}\| : S_{i^*(S_{\leq t})}=0\right]\]
So, pointwise for each realization of $\omega$, we have
\[
 - E\left[\| X_{i^*(\omega)}(\omega) \| : S_{i^*(\omega)}(\omega)\right]  \leq -d_k^{AI}(\infty)
\]
where $i^*(\omega) \equiv i^*(S_{\leq \tau}(\omega))$. 
This inequality is preserved by integrating over $\omega$; thus 
\[
 - E\left[\| X_{i^*(\omega)}(\omega) \|\right] = - E\left[E\left[\| X_{i^*(\omega)}(\omega) \| : S_{i^*(\omega)}(\omega)\right]\right] \leq -d_k^{AI}(\infty)
\]
So for every stopping rule $\tau_{AI}$ in the AI representation regime, we have
\begin{align*}
    -\mathbb{E}_{\tau_{AI}}\left[\| X_{i^*(S_{\leq \tau_{AI}})}\| + c_{\text{AI}}(\tau_{AI}) +  \kappa\right] & \leq -d_k^{\text{AI}}(\infty) - \kappa
\end{align*}
The proof of Theorem \ref{thm:2} established $d_k^{\text{AI}}(\infty) \geq d_k^{\text{IP}}(2)$ for $k$ sufficiently large. Moreover, $\kappa \geq c_{\text{IP}}(2)$ by assumption. Thus
\[-d_k^{\text{AI}}(\infty) - \kappa \leq -d_k^{\text{IP}}(2) - c_{\text{IP}}(2) \quad \mbox{ for $k$ sufficiently large}\]Finally observe that 
\[-d_k^{\text{IP}}(2) - c_{\text{IP}}(2) \leq \sup_{\tau_{\text{IP}}} \mathbb{E}\left[-\min\{\|X_1\|,\|X_2\|,\ldots,\|X_{\tau_{\text{IP}}}\|\} - c_{\text{IP}}(\tau_{\text{IP}})\right] \]
since one feasible stopping rule $\tau_{\text{IP}}$ is to sample twice and stop regardless of the realized history. Thus 
we have the desired result. 
 \end{proof}

\end{document}